\documentclass[a4paper]{amsart}

\usepackage{amsfonts,amsthm,cleveref}

\makeatletter
\newtheorem*{rep@theorem}{\rep@title}
\newcommand{\newreptheorem}[2]{%
\newenvironment{rep#1}[1]{%
 \def\rep@title{#2 \ref{##1}}%
 \begin{rep@theorem}}%
 {\end{rep@theorem}}}
\makeatother

\theoremstyle{plain}
\newtheorem{theorem}{Theorem}
\newtheorem{lemma}[theorem]{Lemma}
\newreptheorem{lemma}{Lemma}
\newtheorem{corollary}[theorem]{Corollary}
\newreptheorem{corollary}{Corollary}

\theoremstyle{definition}
\newtheorem{problem}[theorem]{Problem}
\newtheorem{definition}[theorem]{Definition}

\theoremstyle{remark}

\newcommand{\bra}{\ensuremath{\langle}}
\newcommand{\ket}{\ensuremath{\rangle}}
\newcommand{\Rot}{\ensuremath{\mathrm{Rot}}}
\newcommand{\poly}{\ensuremath{\mathrm{poly}}}
\DeclareMathOperator{\Bin}{Bin}
\DeclareMathOperator{\rank}{rank}

\def\zigzagsym{{\ooalign{\hfil\raise .05ex\hbox{\scriptsize $\mathrm{Z}$}\hfil\crcr\mathhexbox20D}}}
\newcommand{\zigzag}{\ensuremath{\,\zigzagsym\,}}

\title[Explicit amplifiers for outlier correlations]{Explicit correlation amplifiers\\for finding outlier correlations\\in deterministic subquadratic time}

\author{ Matti Karppa
  \and Petteri Kaski
  \and Jukka Kohonen
  \and Padraig~\'{O}~Cath\'{a}in}

\thanks{Helsinki Institute for Information
    Technology HIIT and Department of Computer Science, Aalto
    University, Helsinki, Finland ({\tt firstname.lastname@aalto.fi}).
    The last author's current address is {\tt pocathain@wpi.edu}.}

\thanks{
  This research was funded by the European Research Council,
  under the European Union's Seventh Framework Programme
  (FP/2007-2013) / ERC Grant Agreement 338077 ``Theory and Practice of
  Advanced Search and Enumeration'' (M.K., P.K., J.K.); and Academy of
  Finland, Grants 276031, 282938, 283262 and 283437
  (P.\'{O}\,C.). Work done in part while the second author was
  visiting the Simons Institute for the Theory of Computing.}

\begin{document}


\begin{abstract}
We derandomize G.~Valiant's~[\emph{J.~ACM}~62~(2015)~Art.~13] 
sub\-quadratic-time algorithm for finding outlier correlations in 
binary data. Our derandomized algorithm gives deterministic 
subquadratic scaling essentially for the same parameter range
as Valiant's randomized algorithm, but the precise constants 
we save over quadratic scaling are more modest. 
Our main technical tool for derandomization is an explicit family
of \emph{correlation amplifiers} built via a family of 
zigzag-product expanders by Reingold, Vadhan, and 
Wigderson~[\emph{Ann.~of~Math.}~155~(2002)~157--187]. 
We say that a function $f:\{-1,1\}^d\rightarrow\{-1,1\}^D$ is 
a \emph{correlation amplifier} with 
threshold $0\leq\tau\leq 1$, 
error $\gamma\geq 1$, and 
strength $p$ an even positive integer 
if for all pairs of vectors $x,y\in\{-1,1\}^d$ it holds that 
(i) $|\bra x,y\ket|<\tau d$ implies
$|\bra f(x),f(y)\ket|\leq(\tau\gamma)^pD$; and 
(ii) $|\bra x,y\ket|\geq\tau d$ implies
$\bigl(\frac{\bra x,y\ket}{\gamma d}\bigr)^pD
 \leq\bra f(x),f(y)\ket\leq
 \bigl(\frac{\gamma\bra x,y\ket}{d}\bigr)^pD$.


\end{abstract}

\maketitle


\section{Introduction}

We consider the task of identifying outlier-correlated
pairs from large collections of weakly correlated binary 
vectors in $\{-1,1\}^d$. In more precise terms, we are interested 
in the following computational problem.

\begin{problem}[Outlier correlations]
\label{prob:main}
We are given 
as input two sets $X,Y\subseteq\{-1,1\}^d$ with $|X|=|Y|=n$,
and two thresholds, the \emph{outlier} threshold $\rho>0$ and
the \emph{background} threshold $\tau<\rho$.
Our task is to output all \emph{outlier pairs} $(x,y) \in X \times Y$ 
with $|\bra x,y\ket|\geq\rho d$, subject to the assumption that 
at most $q$ of the pairs $(x,y)\in X\times Y$ satisfy 
$|\bra x,y\ket|>\tau d$.
\end{problem}

\emph{Remark.}  This setting of binary vectors and (Pearson)
correlation is directly motivated, among others, by the connection to
Hamming distance.  Indeed, for two vectors $x,y\in\{-1,1\}^d$ we have
$\bra x,y\ket = d-2D_H(x,y)$, where
$D_H(x,y)=|\{u=1,2,\ldots,d:x(u)\neq y(u)\}|$ is the Hamming distance
between $x$ and $y$.

A na\"\i{}ve way to solve \cref{prob:main} is to compute
the $n^2$ inner products $\bra x,y\ket$ for $(x,y)\in X\times Y$ 
and filter out everything but the outliers. Our interest is in algorithms 
that scale \emph{subquadratically} in $n$, when both $d$ and $q$ are 
bounded from above by slowly growing functions of $n$. That is, we seek 
running times of the form $O(n^{2-\epsilon})$ for a constant 
$\epsilon>0$. Furthermore, we seek to do this without \emph{a~priori} 
knowledge of $q$.

Running times of the form $O(n^{2-c\rho})$ for a constant $c>0$ are
immediately obtainable using techniques such as the seminal
\emph{locality-sensitive hashing} of Indyk and
Motwani~\cite{Indyk1998} and its variants (see \S
\ref{subsec:related}).  However, such algorithms converge to quadratic
running time in $n$ unless $\rho$ is bounded from below by a positive
constant. Our interest is in algorithms that avoid such a ``curse of
weak outliers'' and run in subquadratic time essentially
\emph{independently of the magnitude of $\rho$, provided that $\rho$
  is sufficiently separated from $\tau$}.  Such ability to identify
weak outliers from large amounts of data is useful, among others, in
machine learning from noisy data.

One strategy to circumvent the curse of weak outliers is to pursue
the following intuition: 
(i) partition the input vectors into 
\emph{buckets}~of at most $s$ vectors each, 
(ii) aggregate each bucket into a single vector by taking 
the vector sum, and 
(iii) compute the inner products between the 
$\lceil n/s\rceil \times \lceil n/s\rceil$ pairs of aggregate
vectors. With sufficient separation between $\tau$ and $\rho$, at 
most $q$ of these inner products between aggregates will be large, 
and every outlier pair is discoverable among the at most $s\times s$ 
input pairs that correspond to each large inner product of aggregates. 
Furthermore, a strategy of this form is oblivious to $q$ until 
we actually start searching inside the buckets, which enables 
adjusting $\rho$ and $\tau$ based on the number of large aggregate 
inner products. 

\subsection{Randomized amplification}

Such bucketing strategies have been studied before with the help of
randomization.  In~2012, G.~Valiant~\cite{Valiant2015} 
presented a breakthrough algorithm that, before bucketing, replaces 
each input vector with a randomly subsampled%
\footnote{The dimension is reduced by subsampling because 
the full $d^p$-dimensional Kronecker power is too large to be 
manipulated explicitly to yield subquadratic running times.}{}
 version of its $p^\mathrm{th}$ Kronecker power. Because of the 
tensor-power identity
\begin{equation}
\label{eq:tensorpower}
  \bra x^{\otimes p},y^{\otimes p}\ket =\bra x,y\ket^p\,,
\end{equation}
the ratio between outlier and background correlations gets 
\emph{amplified} to essentially its $p^\mathrm{th}$ power, assuming 
that the sample is large enough so that sufficient concentration 
bounds hold with high probability. This amplification makes the 
outliers stand out from the background even after bucketing, which 
enables detection in subquadratic time using fast matrix 
multiplication. 

A subset of the present authors~\cite{Karppa2016} further improved on
Valiant's algorithm by a modified sampling scheme that 
\emph{simultaneously} amplifies and aggregates the input by further 
use of fast matrix multiplication. With this improvement, 
\cref{prob:main} can be solved in subquadratic time if the 
logarithmic ratio $\log_\tau \rho = (\log\rho)/(\log \tau)$ is 
bounded from above by a constant less than $1$. Also this improved 
algorithm relies on randomization.

\subsection{Explicit amplification}
In this paper we seek \emph{deterministic} subquadratic algorithms.
As with the earlier randomized algorithms, we seek to 
map the $d$-dimensional input vectors to a higher dimension $D$ 
so that inner products are sufficiently amplified in the process. 
Towards this end, we are interested in explicit functions 
$f : \{-1,1\}^d \to \{-1,1\}^D$ that approximate the tensor-power 
identity~\cref{eq:tensorpower}.

\begin{definition}[Correlation amplifier]
  \label{def:ca}
  Let $d$, $D$ and $p$ be positive integers, with $p$ even, and 
  let $0\leq\tau\leq 1$ and $\gamma\geq 1$. 
  A function $f : \{-1,1\}^d \to \{-1,1\}^D$ is a \emph{correlation
    amplifier} with parameters $(d,D,p,\tau,\gamma)$ if 
    for all pairs of vectors $x,y\in\{-1,1\}^d$ we have
  \begin{align}
    \label{eq:coramp1}
    &\text{ if } 
       \bigl|\bra x,y\ket\bigr| < \tau d
    \text{, then } 
       \bigl|\bra f(x),f(y)\ket\bigr| 
       \leq (\tau\gamma)^pD 
    \,\text{;\,\, and} \\
    \label{eq:coramp2}
    &\text{ if } 
       \bigl|\bra x,y\ket\bigr|\geq \tau d 
    \text{, then } 
      \left(
        \!\tfrac{\bra x,y\ket}{\gamma d}\!
      \right)^{\!p}\!\!D 
      \,\leq\,
      \bra f(x),f(y)\ket
      \,\leq\,
      \left(
        \!\tfrac{\gamma \bra x,y\ket}{d}\!
      \right)^{\!p}\!\!D \, .
  \end{align}
\end{definition}

\emph{Remark.} A correlation amplifier $f$ guarantees by 
\cref{eq:coramp1} that correlations below $\tau$ in absolute 
value stay bounded; and by \cref{eq:coramp2} that 
correlations at least $\tau$ in 
absolute value \emph{become positive} and are governed by 
the two-sided approximation  with 
multiplicative error $\gamma\geq 1$. In particular, \cref{eq:coramp2}
implies that correlations at least $\tau$ cannot mask outliers 
under bucketing because all such correlations get positive sign
under amplification.

\medskip
It is immediate that correlation amplifiers exist. For example, 
take $f(x)=x^{\otimes p}$, with $p$ even, to obtain a correlation 
amplifier with $D=d^p$, $\tau=0$, and $\gamma=1$ by~\cref{eq:tensorpower}.
For our present purposes, however, we seek correlation amplifiers 
with $D$ substantially smaller than $d^p$. Furthermore, we seek 
constructions that are \emph{explicit} in the strong%
\footnote{In comparison, a weaker form of explicitness could require,
for example, that there exists a deterministic algorithm that computes
the entire vector $f(x)$ from a given $x$ in time 
$D\cdot \poly(\log D,p)$.}{}
form that there exists a deterministic algorithm that computes any 
individual coordinate of $f(x)$ in time $\poly(\log D,p)$
by accessing $\poly(p)$ coordinates of a given $x\in\{-1,1\}^d$.
In what follows explicitness always refers to this strong form. 

\subsection{Our results}

The main result of this paper is that sufficiently powerful 
explicit amplifiers exist to find outlier correlations in 
deterministic subquadratic time.

\begin{theorem}[Explicit amplifier family]
\label{thm:main}
There exists an explicit correlation amplifier 
$f:\{-1,1\}^{d}\rightarrow\{-1,1\}^{2^K}$ 
with parameters $(d,2^K,2^\ell,\tau,\gamma)$ whenever 
$0<\tau<1$, $\gamma>1$, and $d,K,\ell$ are positive integers with 
\begin{equation}
\label{eq:main-output-dim}
2^K\geq 
d
\biggl(2^{10} \bigr(1-\gamma^{-1/2}\bigl)^{-1}\biggr)^{20\ell+1} 
\biggl(\frac{\gamma}{\tau}\biggr)^{60\,\cdot\,2^{\ell}}
\,.
\end{equation}
\end{theorem}

As a corollary we obtain a deterministic algorithm for finding 
outlier correlations in subquadratic time using bucketing and fast 
matrix multiplication. Let us write $\alpha$ for the limiting 
exponent of rectangular integer matrix multiplication. That is, 
for all constants $\eta>0$ there exists an algorithm that multiplies 
an $m\times\lfloor m^\alpha\rfloor$ integer matrix with 
an $\lfloor m^{\alpha}\rfloor\times m$ integer matrix 
in $O(m^{2+\eta})$ arithmetic operations. 
In particular, it is known that $0.3<\alpha\leq 1$~\cite{LeGall2012}.

\begin{theorem}%
[Deterministic subquadratic algorithm for outlier correlations]
\label{thm:algorithm}
For any constants $0<\epsilon<1$,\, $0<\tau_\mathrm{max}<1$,\, 
$0<\delta<\alpha$, and $C>60$, there exists a deterministic algorithm 
that solves a given instance of \cref{prob:main} in time
\begin{equation}
\label{eq:runtime}
O\biggl(n^{2-\frac{0.99\epsilon(\alpha-\delta)}{4C+1}}+qn^{\delta+\frac{1.99\epsilon(\alpha-\delta)}{4C+1}}\biggr)
\end{equation}
assuming that the parameters $n,d,\rho,\tau$ satisfy the 
following three constraints
\begin{enumerate}
\item
$d\leq n^{\delta}$, 
\item
  $c_1 n^{-c_2} \leq \tau\leq\tau_\mathrm{max}$,
  where $c_1 = \tau_\mathrm{max}^{-\epsilon/100000}$,
  $c_2 = 
\left(1-\frac{0.99\epsilon}{4C+1}\right)\frac{\alpha-\delta}{C}$,
  and
\item
$\log_\tau\rho\leq1-\epsilon$.
\end{enumerate}
\end{theorem}

\emph{Remarks.}
Observe in particular that \cref{eq:runtime} is subquadratic 
regardless of the magnitude of $\rho$ provided that the separation 
between $\rho$ and $\tau$ via $\log_\tau\rho\leq 1-\epsilon$ holds.%
\footnote{%
The technical constraint $c_1n^{-c_2}\leq\tau$
only affects inputs where the dimension $d$ grows essentially 
as a root function of $n$ since $\tau\geq 1/d$.}{}
The constants in \cref{eq:main-output-dim} and \cref{eq:runtime}
have not been optimized beyond our desired goal of obtaining
deterministic subquadratic running time when $d$ and $q$ are bounded
by slowly growing functions of $n$. In particular, \cref{eq:runtime}
gives substantially worse subquadratic running times compared with 
the existing randomized strategies~\cite{Karppa2016,Valiant2015}.  
The algorithm in \cref{thm:algorithm} needs no 
\emph{a priori} knowledge of $q$ and is oblivious to $q$ until 
it starts searching inside the buckets.

\subsection{Overview and discussion of techniques}

A straightforward application of the probabilistic method establishes (\cref{lem:existence})
that low-dimensional correlation amplifiers can be obtained by
subsampling uniformly at random the dimensions of the tensor power
$x^{\otimes p}$ as long as the sample size $D$ is large
enough. 
Thus, in
essence our \cref{thm:main} amounts to \emph{derandomizing} such
a subsampling strategy by presenting an explicit sample that is, up to
the error bounds \cref{eq:coramp1} and \cref{eq:coramp2},
indistinguishable from the ``perfect'' amplifier $x\mapsto x^{\otimes
  p}$ under taking of inner products.

The construction underlying \cref{thm:main} amounts to
an $\ell$-fold composition of explicit \emph{squaring} amplifiers 
($p=2$) with increasingly strong control on the error ($\gamma$) 
and the interval of amplification ($[\tau,1]$) at each successive 
composition.
Towards this end, we require a flexible explicit construction of 
squaring amplifiers with strong control on the error and the interval.
We obtain such a construction from an explicit family of expander 
graphs (\cref{lem:expander-family}) obtainable from 
the explicit zigzag-product constructions of Reingold, Vadhan, and 
Wigderson~\cite{ReingoldVadhanWigderson}. In particular, the 
key to controlling the error and the interval is that the expander
family gives 
\emph{Ramanujan-like}%
\footnote{%
Actual \emph{Ramanujan graphs} (see \cite{HooryLinialWigderson,Lubotzky1988}) 
would give somewhat stronger concentration
$\lambda/\Delta=O(\Delta^{-1/2})$ 
and hence improved constants in \cref{eq:main-output-dim}.
However, we are not aware of a sufficiently fine-grained family
of explicit Ramanujan graphs to comfortably support successive 
squaring.}{}
concentration $\lambda/\Delta\leq 16\Delta^{-1/4}$
of the normalized second eigenvalue $\lambda/\Delta$ by increasing 
the degree $\Delta$.
In essence, since we are working with $\{-1,1\}$-valued vectors,
by increasing the degree we can use the Expander Mixing Lemma 
(\cref{lem:expander-mixing}) and the Ramanujan-like 
concentration to control 
(\cref{lem:approximate-squaring}) how well 
the restriction $x^G$ to the edges of an expander graph $G$ 
approximates the full tensor square $x^{\otimes 2}$ under 
taking of inner products.

Our construction has been motivated by the paradigm of 
\emph{gradually increasing independence}~\cite{Celis2013,Gopalan2015,Gopalan2013,Kane2011} 
in the design of pseudorandom generators. 
Indeed, we obtain the final amplifier gradually by 
successive squarings, taking care that the degree $\Delta_i$ of 
the expander that we apply in each squaring $i=0,1,\ldots,\ell-1$ 
increases with a similar squaring schedule given 
by \cref{eq:tau-i} and \cref{eq:degree-bound}
to simultaneously control the error and the interval, and to 
bound the output dimension roughly by the square of the degree 
of the last expander in the sequence.
Here the term ``gradual'' is not particularly descriptive
since growth under successive squaring amounts to \emph{doubly} 
exponential growth in the number of squarings. 
Yet such growth \emph{can} be seen 
as gradual and controlled in the following sense: we obtain strong amplification 
compared with the final output dimension precisely because 
the first $\ell-1$ squarings ``come for free'' as
$\Delta_0\Delta_1\cdots\Delta_{\ell-2}$ is 
(up to low-order multiplicative terms)
no more than $\Delta_{\ell-1}^2$, essentially because we are 
taking the sum of powers of 2 in the exponent.

The analogy with pseudorandom generators can in fact be pushed 
somewhat further. Namely, a correlation amplifier can be 
roughly seen as a pseudorandom generator that by \cref{eq:coramp2} 
seeks to fool a ``truncated family of uniform combinatorial 
rectangles'' with further control requested by \cref{eq:coramp1} 
below the truncation threshold $\tau$.
To see the rough analogy,
let $z\in\{-1,1\}^d$ be the Hadamard product of the vectors
$x,y\in\{-1,1\}^d$ and observe that \cref{eq:coramp2} seeks 
to approximate (with multiplicative error) the expectation of 
a uniform random entry in the $d^p$-length Kronecker 
power $z^{\otimes p}$ by instead taking the expectation over an 
explicit $D$-dimensional sample given by $f$. 
The Kronecker power $z^{\otimes p}$ is a uniform special case 
(with $z=z_1=z_2=\cdots=z_p$)
of a ``combinatorial rectangle'' formed by a Kronecker product
$z_1\otimes z_2\otimes\cdots\otimes z_p$, and truncation means 
that we only seek approximation in cases where 
$|\sum_{u=1}^d z(u)|\geq\tau d$, and accordingly want constructions
that take this truncation into account---that is, we do not seek 
to fool all combinatorial rectangles and accordingly want 
stronger control on the dimension $D$ (that is, 
the ``seed length'' $\log D$).

For a review of the state of the art in pseudorandom generators
we refer to Gopalan, Kane, and Meka~\cite{Gopalan2015} and
Kothari and Meka~\cite{Kothari2015}.
Our goal to obtain a small output dimension $D$ roughly
corresponds to optimizing the seed length of a pseudorandom generator.

While our explicit construction \cref{eq:main-output-dim} does not
reach the exact output dimension obtainable by \cref{lem:existence},
it should be observed that in our parameter range of interest (with
$\gamma>1$ a constant and $0<\tau\leq\tau_{\mathrm{max}}$ for a
constant $0<\tau_{\mathrm{max}}<1$), both \cref{eq:main-output-dim}
and \cref{eq:prob-output-dim} are of the form $D\geq
d\tau^{-\Theta(p)}$; only the constants hidden by the asymptotic
notation differ between the explicit and nonconstructive bounds.
Moreover, using results of Alon~\cite{Alon2003} we show a \emph{lower
  bound} (\cref{lemma:lowerbound}) on the output dimension $D$ of any
correlation amplifier: namely, that $D \ge \frac{1}{5}
(\frac{1}{\gamma\tau})^p$, when $p$ is in the range governed
by $(\gamma\tau)^p \le 1/100$ and $p \le \frac{(\log
  e)\tau^2d}{8\log(\frac{1}{\gamma\tau})} $.  Thus, viewed as a
pseudorandom generator with ``seed length'' $\log D$, \cref{thm:main}
essentially does not admit improvement except possibly at the
multiplicative constants.

\subsection{Related work and applications}
\label{subsec:related}

\cref{prob:main} is a basic problem in data analysis and
machine learning admitting many extensions, restrictions, and variants.
A large body of work exists studying \emph{approximate near
neighbour search} via techniques such as locality-sensitive hashing
(e.g.~\cite{Andoni2014,Andoni2015,Andoni2016,Indyk1998,Gionis1999,Motwani2007,ODonnell2014}), 
with recent work aimed at derandomization 
(see Pagh~\cite{Pagh2016} and Pham and Pagh~\cite{Pham2016}) 
and resource tradeoffs (see Kapralov~\cite{Kapralov2015})
in particular.
However, these techniques enable subquadratic scaling in $n$ only 
when $\rho$ is bounded from below by a positive constant, whereas 
the algorithm in \cref{thm:algorithm} remains subquadratic 
even in the case of weak outliers when $\rho$ tends to zero 
with increasing $n$, as long as $\rho$ and $\tau$ are separated.
Ahle, Pagh, Razenshteyn, and Silvestri~\cite{Ahle2015} show that
subquadratic scaling in $n$ is not possible for 
$\log_\tau\rho=1-o(1/\sqrt{\log n})$ unless both the Orthogonal
Vectors Conjecture and the Strong Exponential Time 
Hypothesis~\cite{Impagliazzo2001} fail.

In small dimensions, Alman and Williams~\cite{Alman2015} present a
randomized algorithm that finds \emph{exact} Hamming-near neighbours
in a batch-query setting analogous to \cref{prob:main} in subquadratic
time in $n$ when the dimension is constrained to $d=O(\log
n)$. Recently, Chan and Williams~\cite{Chan2016} show how to
derandomize related algorithm designs; also, Alman, Chan and
Williams~\cite{Alman2016} derandomize the probabilistic polynomials
for symmetric Boolean functions used in~\cite{Alman2015}, achieving
deterministic subquadratic batch queries in small dimensions.

One special case of \cref{prob:main} is the problem of 
learning a weight 2 parity function in the presence of noise, or 
\emph{the light bulb problem}.

\begin{problem}[Light bulb problem, L.~Valiant~{\cite{Valiant1988}}]
\label{prob:lightbulb}
Suppose we are given as input a parameter $0<\rho<1$ and 
a set of $n$ vectors in $\{-1,1\}^d$ such that one \textit{planted} 
pair of vectors has inner product at least $\rho d$ in absolute value, 
and all other $n-2$ vectors are chosen independently and uniformly at
random. Our task is to find the planted pair among the $n$ vectors.
\end{problem}

\emph{Remark.} From e.g.~the Hoeffding bound~\cref{eq:hoeffding} 
it follows that there exists a constant $c$ such that when 
$d\geq c\rho^{-2}\log n$ the planted pair is with high probability 
(as $n$ increases) the unique pair in the input with the maximum 
absolute correlation.

For a problem whose instances are drawn from a random ensemble, 
we say that an algorithm solves \emph{almost all} instances 
of the problem if the probability of drawing an instance where 
the algorithm fails tends to zero as $n$ increases.

Paturi, Rajasekaran, and Reif~\cite{Paturi1989},
Dubiner~\cite{Dubiner2010}, and 
May and Ozerov~\cite{May2015} present randomized algorithms
that can be used to solve almost all instances of the light
bulb problem in subquadratic time if we assume that $\rho$ 
is bounded from below by a positive constant; if $\rho$ tends
to zero these algorithms converge to quadratic running time in $n$.

G.~Valiant \cite{Valiant2015} showed that a randomized algorithm
can identify the planted correlation in subquadratic time
on almost all inputs even when $\rho$ tends to zero as $n$ increases. 
As a corollary of \cref{thm:algorithm},
we can derandomize Valiant's design and still retain subquadratic 
running time (but with a worse constant) for almost all inputs, 
except for extremely weak planted correlations with 
$\rho\leq n^{-\Omega(1)}$ that our amplifier is not in general 
able to amplify with sufficiently low output dimension to enable 
an overall subquadratic running time.

\begin{corollary}%
[Deterministic subquadratic algorithm for the light bulb 
problem]
\label{cor:lightbulb}
For any constants $0<\delta<\alpha$, $C>60$, $0<\rho_{\mathrm{max}}<1$,
and $\kappa>1$, there exists a deterministic algorithm that solves 
almost all instances of \cref{prob:lightbulb} in time
\[
O\biggl(n^{2-\frac{0.99(1-1/\kappa)(\alpha-\delta)}{4C+1}}\biggr)
\]
assuming the parameters $n,d,\rho$ satisfy the two constraints
\begin{enumerate}
\item
$5\rho^{-2\kappa}\log n\leq d \leq n^\delta$ and
\item
$c_1n^{-c_2/\kappa}\leq\rho\leq\rho_{\max}$,
\end{enumerate}
where
$c_1 = \rho_\mathrm{max}^{-\kappa\epsilon/100000}$ and
$c_2 = \left(1-\frac{0.99(1-1/\kappa)}{4C+1}\right)\frac{\alpha-\delta}{C}$.
\end{corollary}

\Cref{cor:lightbulb} extends to parity functions of larger (constant)
weight in the presence of noise
(cf.~\cite{Grigorescu2011,Karppa2016,Valiant2015}). This generalized
version of the problem is as follows.

\begin{problem}[Learning parity with noise]
  \label{prob:paritywithnoise}
  Let $S\subseteq [v]$ with $|S| = k$ be the \emph{support} of a
  parity function and $0 < \eta < 1$ the noise level. Our task is to
  determine the set $S$ by drawing independent random examples $(x,y)$
  such that $x\in \{-1,1\}^v$ is chosen uniformly at random, and the
  \emph{label} $y\in \{-1,1\}$ is $y = z\prod_{\ell\in S} x(\ell)$
  where $z \in \{-1,1\}$ is an independent random variable with $\Pr(z
  = -1) = \eta$.
\end{problem}

With no information on $k$, the trivial solution is to enumerate all
$2^v$ subsets of $[v]$ to locate the support $S$. 
Blum, Kalai, and Wasserman~\cite{Blum2003} provide a
non-trivial solution which runs in time and sample complexity
$\text{poly}\bigl(|1-2\eta|^{2^a},2^b\bigr)$ for any positive integers $a,b$
with $ab\geq v$; this is $2^{O(v/\log v)}$ when $\eta\neq 1/2$ is 
a constant independent of $v$. If we assert that $k$ is a constant
independent of $v$, the trivial complexity drops from exponential to $v^k$, 
and non-trivial speed-ups seek to lower the coefficient $1$ of $k$ in the 
exponent. Randomized solutions for constant $k$ include Valiant's 
breakthrough algorithm~\cite{Valiant2015} and our subsequent randomized 
improvement~\cite{Karppa2016} which runs in time 
$\tilde O(v^{\frac{\omega+\epsilon}{3}k}|1-2\eta|^{-\frac{8\omega}{9\epsilon}-\frac 43})$
for any constant $0<\epsilon<\omega/3$. 

Our present contribution is a deterministic algorithm for learning
constant-weight parity functions with noise. Our interest is in the
case where the noise level $\eta$ approaches $1/2$, and accordingly we
assume that $|1-2\eta|$ is bounded from above by a constant less than
$1$. We say that a deterministic algorithm {\em solves almost all
  instances} of \cref{prob:paritywithnoise} if the probability of
drawing an instance on which the algorithm fails tends to zero as $v$
increases.%
\footnote{Observe that from an information-theoretic perspective 
it is a positive-but-negligible-probability event that the drawn examples 
do not uniquely identify $S$.}

\begin{corollary}[Deterministic algorithm for learning parity with noise]
  \label{cor:parities}
  For all constants $0 < \delta < \alpha$, $C>60$, $\xi > 1$, $0<\theta < 1$,
  there exists a constant $k_0$ and a deterministic algorithm
  that for all constants $k\geq k_0$ draws $d$ examples and finds the 
  support of almost all instances of \cref{prob:paritywithnoise} in time
  \begin{equation}
  \label{eq:parity-time-bound}
    O\biggl(
    v^{k\bigl(1-0.245025(\alpha-\delta)^2(1-1/\xi)^2(1+4C)^{-2}\bigr)}
    \biggr)
    \,,
  \end{equation}
  assuming the parameters $v,d,\eta$ satisfy the constraints
  \begin{enumerate}
  \item
  $d \geq \frac{6k}{|1-2\eta|^{2(\xi^2+1)}(1-\theta^{\xi-1})^2}\log v$, and
  \item
    $c_1v^{-c_2\xi^{-2}k/2} \leq |1-2\eta| \leq \theta$,
  \end{enumerate}
  where
  $c_1 = \theta^{-(1-1/\xi)/100000}$
  and
  $c_2 = \left(1-\frac{0.99(1-1/\xi)}{4C+1}\right) \frac{\alpha-\delta}{C}$.
\end{corollary}

Algorithms for learning parity functions enable extensions to 
further classes of Boolean functions such as sparse juntas
and DNFs (cf.~\cite{Feldman2009,Mossel2004,Valiant2015}). 

\section{Preliminaries}

All vectors in this paper are integer-valued. 
For a vector $x\in\mathbb{Z}^d$ we denote the entry 
$u=1,2,\ldots,d$ of $x$ by $x(u)$.
For two vectors $x,y\in\mathbb{Z}^d$ we write 
$\bra x,y\ket=\sum_{u=1}^d x(u)y(u)$ for the inner product 
of $x$ and $y$. We write $\log$ for the logarithm with base $2$
and $\ln$ for the logarithm with base $\exp(1)$.

In our proofs, we need the following
bound due to Hoeffding~\cite[Theorem 2]{Hoeffding1963} which provides
an exponentially small upper bound on the deviation of a sum of
bounded independent random variables from its expectation.

\begin{theorem}[{Hoeffding~\cite[Theorem 2]{Hoeffding1963}}] 
\label{thm:Hoeffding}
  Let $Z_1,Z_2,\ldots,Z_D$ be
  independent random variables satisfying $\ell_i \leq Z_i
  \leq u_i$ for all $1 \leq i \leq D$, and let $Z = \sum_{i=1}^D
  Z_i$. Then, for all $c > 0$, the following holds:
  \begin{equation}
    \label{eq:hoeffding}
    \textrm{\emph{Pr}} \left( Z - \mathrm E[Z] \geq c \right) 
    \leq \exp\left(-\frac{2c^2}{\sum_{i=1}^D(u_i-\ell_i)^2} \right) \, .
  \end{equation}
\end{theorem}

\section{Explicit amplifiers by approximate squaring}

This section proves \cref{thm:main}. We start with 
preliminaries on expanders, show an approximate squaring identity
using expander mixing, and then rely on repeated approximate
squaring for our main construction. The proof is completed
by some routine preprocessing.

\subsection{Preliminaries on expansion and mixing}
We work with undirected graphs, possibly with self-loops and multiple edges. 
A graph $G$ is $\Delta$-\emph{regular} if every vertex
is incident to exactly $\Delta$ edges, with each self-loop (if present) counting 
as one edge. 
Suppose that $G$ is $\Delta$-regular with vertex set $V$, and
let $L$ be a set of $\Delta$ labels such that the 
$\Delta$ edge-ends incident to each vertex have been labeled
with unique labels from $L$. The \emph{rotation map} 
$\Rot_G:V\times L\rightarrow V\times L$ is the 
bijection such that for all $u\in V$ and $i\in L$ we have 
$\Rot_G(u,i)=(v,j)$ if the edge incident to 
vertex $u$ and labeled with $i$ at $u$ leads to the vertex 
$v$ and has the label $j$ at $v$.

For $S,T\subseteq V(G)$, let us write $E(S,T)$ for the set of 
edges of $G$ with one end in $S$ and the other end in $T$. 
Suppose that $G$ has $D$ vertices and let $\lambda_1,\lambda_2,\ldots,\lambda_D$
be the eigenvalues of the adjacency matrix of $G$ with 
$|\lambda_1|\geq |\lambda_2|\geq\cdots\geq|\lambda_D|$. Let us say that a graph 
$G$ is a $(D,\Delta,\lambda)$-\emph{graph} if $G$ has $D$ vertices, 
$G$ is $\Delta$-regular, and $|\lambda_2|\leq\lambda$.
For an excellent survey on expansion and expander graphs, we refer to
Hoory, Linial, and Wigderson~\cite{HooryLinialWigderson}.

\begin{lemma}[Expander mixing lemma, {\cite[Lemma~2.5]{HooryLinialWigderson}}]
\label{lem:expander-mixing}
For all $S,T\subseteq V(G)$ we have
\[
\biggl||E(S,T)|-\frac{\Delta|S||T|}{D}\biggr|
\leq 
\lambda\sqrt{|S||T|}\,.
\]
\end{lemma}

We work with the following family of graphs obtained from the zig-zag product of 
Reingold, Vadhan, and Wigderson~\cite{ReingoldVadhanWigderson}.
In particular \cref{lem:expander-family} gives us 
$\lambda/\Delta\leq 16\Delta^{-1/4}$, which
will enable us to control relative inner products by 
increasing $\Delta$. 

\begin{lemma}
\label{lem:expander-family}
For all integers $t\geq 1$ and $b\geq 10$ 
there exists a $(2^{16bt},2^{4b},16\cdot 2^{3b})$-graph
whose rotation map can be evaluated in time 
$\poly(b,t)$.\footnote{\emph{Caveat.}~Reingold, Vadhan, and Wigderson~\cite{ReingoldVadhanWigderson} work with eigenvalues of the \emph{normalized} adjacency matrix (with $|\lambda_1|=1$) whereas we follow Hoory, Linial, and Wigderson~\cite{HooryLinialWigderson} and work with unnormalized 
adjacency matrices (with $|\lambda_1|=\Delta$) in the manuscript proper. \cref{appendix:expanders} works with normalized adjacency matrices for compatibility with Reingold, Vadhan, and Wigderson~\cite{ReingoldVadhanWigderson}.}{}
\end{lemma}
\begin{proof}
See \cref{appendix:expanders}. 
\end{proof}

\subsection{Main construction}

The main objective of this section is to prove the following lemma,
which we will then augment to \cref{thm:main} by routine
preprocessing of the input dimension. 

\begin{lemma}[Repeated approximate squaring]
\label{lem:main}
There exists an explicit correlation amplifier 
$\hat f:\{-1,1\}^{2^k}\rightarrow\{-1,1\}^{2^K}$ 
with parameters $(2^k,2^K,2^\ell,\tau_0,\gamma_0)$ 
whenever $0<\tau_0<1$, $\gamma_0>1$, and $k,K,\ell$ are positive 
integers with
\begin{equation}
\label{eq:output-dim}
2^K\geq 
2^k \biggl(2^{10}\bigr(1-\gamma_0^{-1}\bigl)^{-1}\biggr)^{20\ell} 
\biggl(\frac{\gamma_0}{\tau_0}\biggr)^{40\,\cdot\,2^{\ell}-20}
\,.
\end{equation}
\end{lemma}

\paragraph{Approximate squaring via expanders}
For a vector $x\in\{-1,1\}^D$, let us write $x^{\otimes 2}\in\{-1,1\}^{D^2}$ 
for the Kronecker product of $x$ with itself. Our construction for correlation amplifiers will rely
on approximating the \emph{squaring identity}
\[
\bra x^{\otimes 2},y^{\otimes 2}\ket=
\bra x,y\ket^2\,,
\]
for vectors in $\{-1, 1\}^{D}$. 
In more precise terms, let $G$ be a $(D,\Delta,\lambda)$-graph 
and let $x^G\in\{-1,1\}^{\Delta D}$ be a vector that contains 
each coordinate $x(u)x(v)$ of $x^{\otimes 2}$ with 
$(u,v)\in V(G)\times V(G)$ exactly once for each edge of $G$ 
that joins the vertex $u$ to the vertex $v$. 
Equivalently, let $\Rot_G:V\times L\rightarrow V\times L$ be a
rotation map for $G$, and define $x^G$ for all $u\in V$ and all 
$i\in L$ by $x^G(u,i)=x(u)x(v)$ where $v\in V$ is given by 
$\Rot_G(u,i)=(v,j)$.
In particular, $x^G$ has exactly $\Delta D$ coordinates.

\begin{lemma}[Approximate squaring]
\label{lem:approximate-squaring}
For all $x,y\in \{-1, 1\}^D$ we have
\[ 
\left| \bra x^G,y^G \ket-\frac{\Delta}{D}\bra x^{\otimes 2},y^{\otimes 2}\ket \right| 
\leq 
2\lambda D\,. 
\]
\end{lemma}

\begin{proof}
Let $S=\{u\in V(G):x(u)=y(u)\}$ and let us write $\bar S=V(G)\setminus S$.
Since $x,y$ are $\{-1,1\}$-valued, we have
\[
\bra x^G, y^G \ket = 
|E(S,S)|+|E(\bar S,\bar S)|-|E(S,\bar S)|-|E(\bar S,S)|\,. 
\]
Observing that 
\[
|S|^2+|\bar S|^2-|S||\bar S|-|\bar S||S|
=\bigl(2|S|-D\bigr)^2
=\bra x,y\ket^2
=\bra x^{\otimes2},y^{\otimes2}\ket
\]
and applying \cref{lem:expander-mixing} four times, we have
\[ 
\left|\bra x^G,y^G \ket-\frac{\Delta}{D}\bra x^{\otimes 2},y^{\otimes 2}\ket\right| 
\leq 
\lambda \bigl(D+2\sqrt{|S|(D-|S|)}\bigr) \leq 2\lambda D\,. 
\] 
\end{proof}

\paragraph{The amplifier function}
We now construct an amplifier function $\hat f$ that uses $\ell$ 
approximate squarings, $\ell\geq 1$, with the graphs drawn from 
the graph family in \cref{lem:expander-family}.
Accordingly, we assume that all vectors have lengths that are positive integer powers of $2$.

The input $x=\tilde x_0\in\{-1,1\}^{d_0}$ to the amplifier has dimension $d_0=2^k$ for a positive
integer $k$. For $i=0,1,\ldots,\ell-1$, suppose we have the vector $\tilde x_i\in\{-1,1\}^{d_i}$. 
Let $b_i$ be a positive integer whose value will be fixed later. 
Let $t_i$ be the unique positive integer with 
\[
d_i\leq D_i=2^{16b_it_i}<2^{16b_i}d_i\,.
\]
Note in particular that $d_i$ divides $D_i$ since $d_i$ is 
a power of 2. 
Let $G_i$ be a $(2^{16b_it_i}, 2^{4b_i},$ $16\cdot 2^{3b_i})$-graph 
from \cref{lem:expander-family}.
Take $D_i/d_i$ copies of $\tilde x_i$ to obtain the 
vector $x_i\in\{-1,1\}^{D_i}$. 
Let $\tilde x_{i+1}=x_i^{G_i}\in\{-1,1\}^{d_{i+1}}$ with 
$d_{i+1}=\Delta_iD_i$ and $\Delta_i=2^{4b_i}$. 
The amplifier outputs $\hat f(x)=\tilde x_\ell$ with 
$\tilde x_\ell\in\{-1,1\}^{d_\ell}$. 

Since the graph family in \cref{lem:expander-family} 
admits rotation maps that can be computed in time $\poly(b,t)$,
we observe that $\hat f$ is explicit. Indeed, from
the construction it is immediate that to compute
any single coordinate of $\hat f(x)$ it suffices to 
(i) perform in total $2^{\ell-1-i}$ evaluations of 
the rotation map of the graph $G_i$ for each $i=0,1,\ldots,\ell-1$,
and (ii) access at most $2^\ell$ coordinates of $x$.
Since $b_it_i=O(\log d_\ell)$ for all $i=0,1,\ldots,\ell-1$, we 
have that we can compute any coordinate of $\hat f(x)$ in time
$\poly(\log d_\ell,2^\ell)$ and accessing at most $2^\ell$ 
coordinates of $x$.

\paragraph{Parameterization and analysis}
Fix $\tau_0>0$ and $\gamma_0>1$. To parameterize the amplifier (that is, it remains to fix the values $b_i$), let
us track a pair of vectors as it proceeds through the $\ell$ approximate squarings for $i=0,1,\ldots,\ell-1$.

We start by observing that copying preserves \emph{relative} inner 
products. That is, for any pair of vectors 
$\tilde x_i,\tilde y_i\in\{-1,1\}^{d_i}$ we have 
$\bra \tilde x_i,\tilde y_i\ket=\nu_i d_i$ if and only if 
$\bra x_i, y_i \ket=\nu_i D_i$ for $0\leq\nu_i\leq 1$. 

An easy manipulation of \cref{lem:approximate-squaring} using 
the parameters in \cref{lem:expander-family} gives us additive 
control over an approximate squaring via
\begin{equation}
\label{eq:squaring-additive}
\nu_i^2 - 32\Delta_i^{-1/4} \leq \nu_{i+1} \leq \nu_i^2 + 32\Delta_i^{-1/4}\,. 
\end{equation}
For all inner products that are in absolute value above a threshold, we want to turn 
this additive control into multiplicative control via
\begin{equation}
\label{eq:squaring-multiplicative}
\nu_i^2\gamma_0^{-1} \leq \nu_{i+1} \leq \nu_i^2\gamma_0\,. 
\end{equation}
Let us insist this multiplicative control holds whenever $|\nu_i|\geq\tau_i$ for 
the threshold parameter $\tau_i$ defined for all $i=0,1,\ldots,\ell-1$ by 
\begin{equation}
\label{eq:tau-i}
\tau_{i+1}=\gamma_0^{-1}\tau_i^2\,.
\end{equation}
Enforcing \cref{eq:squaring-multiplicative} via 
\cref{eq:squaring-additive} at the threshold, let us assume that 
\begin{equation}
\label{eq:delta-i}
\tau_i^2\gamma_0^{-1}\leq\tau_i^2-32\Delta_i^{-1/4}\,.
\end{equation}
The next lemma confirms that assuming \cref{eq:delta-i} gives two-sided control 
of inner products which is retained to the next approximate squaring. The following 
lemma shows that small inner products remain small. 

\begin{lemma}
\label{lem:large-squaring}
If $\tau_i\leq |\nu_i|$, then $\nu_i^2\gamma_0^{-1}\leq\nu_{i+1}\leq\nu_i^2\gamma_0$ and $\tau_{i+1}\leq\nu_{i+1}$.
\end{lemma}
\begin{proof}
From \cref{eq:squaring-additive} and \cref{eq:delta-i}, we have
\begin{equation}
\label{eq:nu-step-2}
\left|\nu_{i+1}-\nu_i^2\right|
\leq 32\Delta_i^{-1/4}
\leq (1-\gamma_0^{-1})\tau_i^2
\leq (1-\gamma_0^{-1})\nu_i^2\,.
\end{equation}
Observe that $1-\gamma_0^{-1}\leq \gamma_0-1$. Thus, 
from \cref{eq:nu-step-2} we conclude that 
\[
\nu_{i+1}
\leq \nu_i^2+(1-\gamma_0^{-1})\nu_i^2
\leq \nu_i^2+(\gamma_0-1)\nu_i^2
= \gamma_0\nu_i^2\,.
\]
In the converse direction, from \cref{eq:nu-step-2} 
and \cref{eq:tau-i} we conclude that
\[
\nu_{i+1}
\geq \nu_i^2-(1-\gamma_0^{-1})\nu_i^2
\geq \gamma_0^{-1}\nu_i^2
\geq \gamma_0^{-1}\tau_i^2
=\tau_{i+1}\,.
\]
\end{proof}

\begin{lemma}
\label{lem:small-squaring}
If $|\nu_i|<\tau_i$, then $|\nu_{i+1}|\leq\tau_i^2\gamma_0$.
\end{lemma}
\begin{proof}
From \cref{eq:squaring-additive} and \cref{eq:delta-i}, we have
\begin{equation}
\label{eq:nu-step-1}
\left|\nu_{i+1}-\nu_i^2\right|
\leq 32\Delta_i^{-1/4}
\leq (1-\gamma_0^{-1})\tau_i^2\,.
\end{equation}
Since $1-\gamma_0^{-1}\leq \gamma_0-1$, from \cref{eq:nu-step-1} 
we conclude that 
\[
|\nu_{i+1}|
\leq \nu_i^2+(1-\gamma_0^{-1})\tau_i^2
\leq \tau_i^2+(\gamma_0-1)\tau_i^2
= \gamma_0\tau_i^2\,.
\]
\end{proof}

Let us now make sure that \cref{eq:delta-i} holds.
Solving for $\Delta_i$ in \cref{eq:delta-i}, we have 
\begin{equation}
\label{eq:degree-bound}
\Delta_i\geq\left(32(1-\gamma_0^{-1})^{-1}\tau_{i}^{-2}\right)^{4}\,.
\end{equation}
In particular, we can make sure that \cref{eq:degree-bound} and hence \cref{eq:delta-i}
holds by simply choosing a large enough $\Delta_i$ (that is, a large enough $b_i$). 

Before proceeding with the precise choice of $b_i$ for $i=0,1,\ldots,\ell-1$, let us analyze the
input--output relationship of the amplifier $\hat f$ using \cref{lem:large-squaring}
and \cref{lem:small-squaring}. Let $x,y\in\{-1,1\}^{d_0}$ be two
vectors given as input with $\bra x,y\ket=\nu_0d_0$. The outputs
$\hat f(x),\hat f(y)\in\{-1,1\}^{d_\ell}$ then satisfy $\bra \hat f(x),\hat f(y)\ket=\nu_\ell d_\ell$,
where the following two lemmas control $\nu_\ell$ via $\nu_0$.

\begin{lemma}
\label{lem:large-amplify}
If $|\nu_0|\geq\tau_0$, then 
$\nu_0^{2^\ell}\gamma_0^{-2^\ell+1}\leq\nu_\ell\leq\nu_0^{2^\ell}\gamma_0^{2^\ell-1}$.
\end{lemma}
\begin{proof}
Use induction on $i$, where \cref{lem:large-squaring} 
gives the inductive step.
\end{proof}

\begin{lemma}
\label{lem:small-amplify}
If $|\nu_0|<\tau_0$, then $|\nu_\ell|\leq\tau_0^{2^\ell}\gamma_0^{2^\ell-1}$.
\end{lemma}
\begin{proof}
From \cref{eq:tau-i} we have $\tau_i=\tau_0^{2^i}\gamma_0^{-2^i+1}$.
Let us show by induction on $i$ that $|\nu_i|\leq\tau_0^{2^i}\gamma_0^{2^i-1}$. 
The base case $i=0$ is immediate. For $i\geq 1$, there are two
cases to consider. First suppose that $|\nu_i|<\tau_i$. Then,
by \cref{lem:small-squaring} we have 
$|\nu_{i+1}|\leq\tau_i^2\gamma_0\leq\tau_0^{2^{i+1}}\gamma_0^{-2^{i+1}+3}\leq\tau_0^{2^{i+1}}\gamma_0^{2^{i+1}-1}$ 
since $\gamma_0>1$. Next suppose that $|\nu_i|\geq\tau_i$. Then,
by \cref{lem:large-squaring} we have 
$|\nu_{i+1}|\leq\nu_i^2\gamma_0\leq\tau_0^{2^{i+1}}\gamma_0^{2^{i+1}-1}$.
\end{proof}

Since $\gamma_0>1$, from \cref{lem:large-amplify} and \cref{lem:small-amplify} 
it now follows that $\hat f$ meets the required amplification 
constraints \cref{eq:coramp1} and \cref{eq:coramp2} with $p=2^\ell$,
$\tau=\tau_0$, and $\gamma=\gamma_0$.

Let us now complete the parameterization and derive an upper bound for $d_\ell$. 
For each $i=0,1,\ldots,\ell-1$, take $b_i$ to be the smallest 
nonnegative integer so that $b_i\geq 10$ and $\Delta_i=2^{4b_i}$ 
satisfies \cref{eq:degree-bound}. 
Since $D_i\leq 2^{16b_i}d_i=\Delta_i^4d_i$, we have
$d_{i+1} = \Delta_{i}D_{i} \leq \Delta_{i}^{5}d_{i}$, and hence
\[
d_{\ell} \leq \left(\Delta_{\ell-1}\Delta_{\ell-2}\cdots\Delta_0\right)^5\!d_{0}\,.
\]
Recall that $d_0=2^k$. From \cref{eq:degree-bound} we have that 
\[
\Delta_i
=2^{4b_i}
\leq \max\bigl(2^{40},2^4\bigl(32(1-\gamma_0^{-1})^{-1}\tau_i^{-2}\bigr)^4\bigr)
\leq \bigl(2^{10}(1-\gamma_0^{-1})^{-1}\tau_i^{-2}\bigr)^4
\,.
\] 
Since $\tau_i=\tau_0^{2^i}\gamma_0^{-2^i+1}$ by \cref{eq:tau-i}, it follows that 
\[
d_{\ell}\leq 
2^k \biggl(2^{10}\bigr(1-\gamma_0^{-1}\bigl)^{-1}\biggr)^{20\ell} 
\biggl(\frac{\gamma_0}{\tau_0}\biggr)^{20(2^{\ell+1}-1)}
\,.
\]
Repeatedly taking two copies of the output as necessary, for all $2^K$ with 
$2^K\geq d_{\ell}$ we obtain a correlation amplifier with parameters 
$(2^k,2^K,2^{\ell},\tau_0,\gamma_0)$. This completes the proof
of \cref{lem:main}. \hfill $\qed$

\subsection{Copy-and-truncate preprocessing of the input dimension}

We still want to remove the assumption from \cref{lem:main} 
that the input dimension is a positive integer power of 2.
The following copy-and-truncate preprocessing will be sufficient
towards this end.

Let $x\in\{-1,1\}^d$ and let $k$ be a positive integer. 
Define the vector $\hat x\in\{-1,1\}^{2^k}$ by concatenating
$\lceil 2^k/d\rceil$ copies of $x$ one after another, and 
truncating the result to the $2^k$ first coordinates to
obtain $\hat x$.%

Let us study how the map $x\mapsto\hat x$ operates on a pair
of vectors $x,y\in\{-1,1\}^d$. For notational compactness, 
let us work with relative inner products $\nu,\hat\nu$ 
with $\bra x,y\ket=\nu d$ and $\bra \hat x,\hat y\ket=\hat\nu 2^k$.

\begin{lemma}
\label{lem:copy-and-truncate}
For any $0<\tau_0<1$, $\gamma_0>1$, and 
$2^k\geq 2d\tau_0^{-1}(1-\gamma_0^{-1})^{-1}$
we have that
\begin{enumerate}
\item
$|\nu|<\tau_0$ implies $|\hat\nu|\leq\gamma_0\tau_0$, 
\item
$\left|\nu\right|\geq \tau_{0}$ implies $\gamma_0^{-1}\nu\leq\left|\hat\nu\right|\leq\gamma_0\nu$.
\end{enumerate}
\end{lemma}
\begin{proof}
Let $\ell $ and $t$ be the unique integers such that 
$2^k+\ell=td$ with $0\leq \ell<d$. Since we are leaving
out $\ell$ coordinates, we have 
\[
2^{-k}(\nu td-\ell)\leq\hat\nu\leq 2^{-k}(\nu td+\ell)\,.
\]
Suppose that $|\nu|<\tau_0$. We
have 
\[
|\hat\nu|
\leq 2^{-k}\bigl(|\nu|td+\ell\bigr)
\leq 2^{-k}\bigl(|\nu|2^k+2\ell\bigr)
\leq \tau_0+2^{1-k}d\,.
\]
Observe that $1-\gamma_0^{-1}\leq\gamma_0-1$. Since by hypothesis
\[
2^k\geq 
2d\tau_0^{-1}(1-\gamma_0^{-1})^{-1}\geq
2d\tau_0^{-1}(\gamma_0-1)^{-1}\,, 
\]
we thus have $|\hat\nu|\leq\gamma_0\tau_0$.

For $\nu\geq\tau_0$ we have
\[
\nu-2^{-k}d\leq 2^{-k}(\nu 2^k-d)\leq\hat\nu\leq 2^{-k}(\nu2^k+2d)\leq \nu+2^{1-k}d\,.
\]
Similarly, for $\nu\leq-\tau_0$ we have
\[
\nu-2^{1-k}d\leq 2^{-k}(\nu 2^k-2d)\leq\hat\nu\leq 2^{-k}(\nu2^k+d)\leq \nu+2^{-k}d\,.
\]
By hypothesis we have
\[
2^k\geq 
2d\tau_0^{-1}\max\bigl((1-\gamma_0^{-1})^{-1},(\gamma_0-1)^{-1}\bigr)\,.
\]
Thus we have both
$\gamma_0^{-1}\nu\leq\hat\nu\leq\gamma_0\nu$ if $\nu\geq\tau_0$,
and $\gamma_0\nu\leq\hat\nu\leq\gamma_0^{-1}\nu$ if $\nu\leq-\tau_0$.
\end{proof}

\subsection{Completing the proof of \cref{thm:main}} 

Let $d,K,\ell,\tau,\gamma$ be parameters meeting the constraints
in \cref{thm:main}, in particular the constraint
\cref{eq:main-output-dim}. To construct a required amplifier $f$, 
we preprocess each input vector $x$ with copy-and-truncate, 
obtaining a vector $\hat x$ of length $2^{k}$. We then 
then apply an amplifier 
$\hat f:\{-1,1\}^{2^k}\rightarrow\{-1,1\}^{2^K}$ 
given by \cref{lem:main}. In symbols, we define 
$f:\{-1,1\}^{d}\rightarrow\{-1,1\}^{2^K}$
for all $x\in\{-1,1\}^d$ by $f(x)=\hat f(\hat x)$.
It is immediate from \cref{lem:main} and 
\cref{lem:copy-and-truncate} that the resulting composition
is explicit. 

We begin by relating the given parameters of \cref{thm:main} 
to those of \cref{lem:main}. Take $\gamma_0=\gamma^{1/2}$, 
$\tau_0=\tau\gamma^{-1}$, and select the minimal value of $k$ so 
that the constraint in \cref{lem:copy-and-truncate} is satisfied; 
that is $2^{k}$ is constrained as follows,
\[ 2d (1-\gamma^{-1/2})^{-1} \gamma \tau^{-1} \leq 2^{k} < 4d (1-\gamma^{-1/2})^{-1} \gamma \tau^{-1}\,.\]
Substituting this upper bound into the bound of \cref{lem:main}, we get a lower bound for $2^{K}$, 
\begin{equation}
\label{eq:upper-fairly-sharp} 
2^{K} \geq 2^{-8} d 
\left( 2^{10} (1- \gamma^{-1/2})^{-1} \right)^{20 \ell + 1} 
\frac{\gamma}{\tau}\left(\frac{\gamma^{60}}{\tau^{40}}\right)^{2^{\ell}}\frac{\tau^{20}}{\gamma^{30}} \,.
\end{equation}

Observe that an integer $2^K$ satisfying 
\cref{eq:main-output-dim} also satisfies 
\cref{eq:upper-fairly-sharp}. We have not attempted 
to optimise our construction, and prefer the 
the statement of \cref{thm:main}
as it is reasonably clean and is
sufficient to prove \cref{thm:algorithm}. 

Let us study how the map $x\mapsto f(x)$ operates on a pair
of vectors $x,y\in\{-1,1\}^d$. For notational compactness, 
again we work with relative inner products $\nu,\hat\nu,\phi$ 
with $\bra x,y\ket=\nu d$, $\bra \hat x,\hat y\ket=\hat\nu 2^k$,
and $\bra f(x),f(y)\ket=\phi 2^K$. Observe that in the notation 
of the proof of \cref{lem:main}, we have $\hat\nu=\nu_0$ 
and $\phi=\nu_\ell$. 

\begin{lemma}
\label{lem:cor-amp-step1}
If $\left| \nu \right| < \tau$ then $\left|\phi \right| \leq (\gamma\tau)^{2^{\ell}}$. 
\end{lemma}
\begin{proof}
First we show that $\left| \hat{\nu}\right| \leq \gamma_{0} \tau$, 
dividing into cases as in Lemma~\ref{lem:copy-and-truncate}. 
If $\left| \nu \right| < \tau_{0}$ then 
$\left| \hat{\nu} \right| < \gamma_{0}\tau_{0} = 
\gamma_{0}^{-1} \tau \leq \gamma_{0} \tau$. 
If $\tau_{0} \leq \nu < \tau$ then $\hat{\nu} \leq \gamma_{0} \nu \leq \gamma_{0}\tau$. 
If $-\tau < \nu \leq \tau_{0}$ then $\hat{\nu} \geq \gamma_{0} \nu \geq -\gamma_{0}\tau$. 

To complete the proof, we condition on $\left| \hat{\nu} \right|$. 
If $\left| \hat{\nu} \right| \leq \tau_{0}$ then Lemma~\ref{lem:small-amplify} 
applies, and we have 
\[ \left| \phi \right| = \left| \nu_{\ell} \right| \leq \tau_{0}^{2^{\ell}}\gamma_{0}^{2^{\ell}-1} < (\tau\gamma)^{2^{\ell}} \,.\] 
Otherwise, $\tau_{0} \leq \left| \hat{\nu} \right| < \tau$ and by Lemma~\ref{lem:large-amplify} 
we have 
 \[ 0 < \phi = \nu_{\ell} \leq \nu_{0}^{2^{\ell}} \gamma_{0}^{2^{\ell}-1} 
 \leq \tau^{2^{\ell}} \gamma_{0}^{2^{\ell}-1} \leq (\tau\gamma)^{2^{\ell}} \,. \]
\end{proof}

\begin{lemma}
\label{lem:cor-amp-step2}
If $\left| \nu \right| \geq \tau$ then $(\nu\gamma^{-1})^{2^{\ell}} \leq \phi \leq (\nu\gamma)^{2^{\ell}}$. 
\end{lemma} 
\begin{proof}
It will be convenient to split the analysis according to whether $\nu$ is positive or negative.
Suppose first that $\nu \geq \tau$.

Then by Lemma~\ref{lem:copy-and-truncate} we have that
\begin{equation}
\label{eq:hat-relation}
\gamma_{0}^{-1} \nu \leq \hat{\nu} \leq \gamma_{0} \nu\,.
\end{equation}
Since $\hat{\nu} \geq \nu\gamma_{0}^{-1} \geq \tau\gamma_{0}^{-1} = \tau_{0}\gamma_{0}\geq \tau_{0}$, Lemma~\ref{lem:large-amplify} applies, yielding
\[ \hat{\nu} \gamma_{0}^{-2^{\ell} + 1} \leq \nu_{\ell} \leq \hat{\nu}^{2^{\ell}} \gamma_{0}^{2^{\ell} -1}\,. \]
Now, we substitute $\phi = \nu_{\ell}$ and bound $\hat{\nu}$ as in \cref{eq:hat-relation},
\[
\left( \nu \gamma_{0}^{-1} \right)^{2^{\ell}} \gamma_{0}^{-2^{\ell} + 1}
\leq \phi \leq
\left( \nu \gamma_{0}^{-1} \right)^{2^{\ell}} \gamma_{0}^{2^{\ell} - 1}\,.\]
Substituting $\gamma = \gamma_{0}^{1/2}$ and observing that $\gamma \geq 1$
provides the required bound
\[
\left( \nu \gamma^{-1} \right)^{2^{\ell} }
\leq \phi \leq
\left( \nu \gamma \right)^{2^{\ell} }\,.\]

The case that $\nu \leq -\tau$ essentially follows from multiplying
all inequalities in the positive case by $-1$.
\end{proof} 

Now, $f$ satisfies \cref{eq:coramp1} and \cref{eq:coramp2} with $p =
2^{\ell}$ by \cref{lem:cor-amp-step1,lem:cor-amp-step2}
respectively. This completes the proof of \cref{thm:main}. \hfill$\qed$

\section{A deterministic algorithm for outlier correlations}

This section proves \cref{thm:algorithm}. We start by 
describing the algorithm, then parameterize it and establish 
its correctness, and finally proceed to analyze the running time.

\subsection{The algorithm}

Fix the constants $\epsilon,\tau_\mathrm{max},\delta,C$ as in
\cref{thm:algorithm}.
Based on these constants, fix the constants $0<\sigma<1$ and 
$\gamma>1$. (We fix the precise values of $\sigma$ and $\gamma$ 
later during the analysis of the algorithm, and stress that 
$\sigma,\gamma$ do not depend on the given input.)

Suppose we are given as input the parameters $0<\tau<\rho<1$
and $X,Y\subseteq\{-1,1\}^d$ with $|X|=|Y|=n$ so that the requirements
in \cref{thm:algorithm} hold. 
We work with a correlation amplifier
$f:\{-1,1\}^d\rightarrow\{-1,1\}^D$ with parameters 
$(d,D,p,\tau,\gamma)$. (We fix the precise values of the parameters 
$p$ and $D$ later during the analysis of the algorithm so that
$f$ originates from \cref{thm:main}.)

The algorithm proceeds as follows.
First, apply $f$ to each vector in $X$ and $Y$ to obtain the
sets $X_f$ and $Y_f$.  
Let $s=\lfloor n^{\sigma}\rfloor$.
Second, partition the $n$ vectors in both $X_f$ and $Y_f$ into 
$\lceil n/s\rceil$ buckets of size at most $s$ each, and take 
the vector sum of the vectors in each bucket to obtain the
sets $\tilde X_f,\tilde Y_f\subseteq\{-s,-s+1,\ldots,s-1,s\}^D$
with $|\tilde X_f|,|\tilde Y_f|\leq\lceil n/s\rceil$. 
Third, using fast rectangular matrix
multiplication on $\tilde X_f$ and $\tilde Y_f$, 
compute the matrix $Z$ whose entries
are the inner products $\bra\tilde x,\tilde y\ket$ for 
all $\tilde x\in\tilde X_f$ and all $\tilde y\in\tilde Y_f$.
Fourth, iterate over the entries of $Z$, and whenever 
the \emph{detection inequality}
\begin{equation}
\label{eq:detect}
\bra\tilde x,\tilde y\ket>n^{2\sigma}(\tau\gamma)^p D
\end{equation}
holds, brute-force search for outliers among the at most 
$s^2$ inner products in the corresponding pair of buckets.
Output any outliers found.

\subsection{Parameterization and correctness}
\label{subsec:correctness}

Let us now parameterize the algorithm and establish its correctness.
Since $\gamma>1$ is a constant and assuming that $p$ is large 
enough, by \cref{thm:main} we can select $D$ to be
the integer power of $2$ with
\[
\frac{1}{2}d\biggl(\frac{\gamma}{\tau}\biggr)^{Cp}
< D\leq 
d\biggl(\frac{\gamma}{\tau}\biggr)^{Cp}\,.
\]
Recall that we write $\alpha$ for the exponent of rectangular 
matrix multiplication. To apply fast rectangular matrix 
multiplication in the third step of the algorithm, we want
\begin{equation}
\label{eq:D-upper}
D\leq 2\biggl(\frac{n}{s}\biggr)^\alpha\,,
\end{equation}
so recalling that $d\leq n^\delta$ and $n^\sigma-1<s$,
it suffices to require that
\[
\biggl(\frac{\gamma}{\tau}\biggr)^{Cp}
\leq n^{(1-\sigma)\alpha-\delta}\,.
\]
Let us assume for the time being that $(1-\sigma)\alpha-\delta>0$.
(We will justify this assumption later when we choose a value
for $\sigma$.) 
Let $p$ be the unique positive-integer power of $2$ such that
\begin{equation}
\label{eq:p-lower}
\frac{((1-\sigma)\alpha-\delta)\log n}{2C\log\frac{\gamma}{\tau}}<
p\leq \frac{((1-\sigma)\alpha-\delta)\log n}{C\log\frac{\gamma}{\tau}}\,.
\end{equation}
We will later, when fixing $\sigma$ and $\gamma$, make sure that
the right-hand side in \cref{eq:p-lower} is at least~$1$, so that $p$
exists and is positive.

Let us now consider a single entry $\bra\tilde x,\tilde y\ket$ in~$Z$,
and analyze how the corresponding (at most $s^2$) inner products $\bra
x,y \ket$ between the two buckets of input vectors relate to the
detection inequality \cref{eq:detect}.  We make two claims:

Claim 1 (background case).  If all of the inner products have $|\bra x,y
\ket| \le \tau d$, then \cref{eq:detect} \emph{does not} hold, so the
algorithm will not search inside this pair of buckets.  This claim
will be used to control the running time.  The claim follows directly
from \cref{eq:coramp1} and \cref{eq:coramp2}, since there are at most
$s^2 \le n^{2\sigma}$ inner products, each having $|\bra
f(x),f(y)\ket| \le (\tau\gamma)^p D$.

Claim 2 (outlier case).  If at least one of the inner products has $|\bra
x,y \ket| \ge \rho d$, then \cref{eq:detect} \emph{holds}, so the
algorithm searches inside this pair of buckets.  This guarantees that the
outliers are detected.

Note that in the third case, namely, if some inner products have
$|\bra x,y \ket| > \tau d$ but none has $|\bra x,y \ket| \ge \rho d$,
we make no claim on whether \eqref{eq:detect} holds or not.  The
algorithm is not required to search inside such pairs of buckets
(since there are no outliers there), but may so do without hindering
our overall running time bound.

We proceed to parameterize the algorithm so that Claim~2 holds.  In
the outlier case, by \cref{eq:coramp1} and \cref{eq:coramp2}, there is
at least one inner product with $\bra f(x),f(y) \ket \ge (\rho
\gamma^{-1})^p D$, and the remaining at most $n^{2\sigma}$ inner
products have $\bra f(x),f(y) \ket \ge -(\tau\gamma)^p D$.  Thus in
the outlier case we have
\begin{equation}
  \label{eq:outliers-signalled}
  \bra \tilde x, \tilde y \ket \ge (\rho\gamma^{-1})^pD - n^{2\sigma}(\tau\gamma)^pD.
\end{equation}
For Claim~2 we need the detection inequality \cref{eq:detect} to hold
whenever \cref{eq:outliers-signalled} holds.  Towards this end, it
suffices to require that
\[
  \bigl(\rho\gamma^{-1}\bigr)^p-n^{2\sigma}\bigl(\tau\gamma\bigr)^p
  >
  n^{2\sigma}\bigl(\tau\gamma\bigr)^p\,.
\]
Rearranging and solving for $p$, we require that
\begin{equation}
\label{eq:p-need}
p>\frac{1+2\sigma\log n}{\log\frac{\rho}{\tau\gamma^2}}\,.
\end{equation}
From \cref{eq:p-lower} and \cref{eq:p-need} we thus see
that it suffices to have
\[
p>\frac{((1-\sigma)\alpha-\delta)\log n}{2C\log\frac{\gamma}{\tau}}
\geq
\frac{1+2\sigma\log n}{\log\frac{\rho}{\tau\gamma^2}}\,,
\]
or equivalently,
\begin{equation}
\label{eq:sigma-require}
\frac{\log\frac{\rho}{\tau\gamma^2}}{\log\frac{\gamma}{\tau}}
\geq
\frac{\frac{2C}{\log n}+4C\sigma}{(1-\sigma)\alpha-\delta}\,.
\end{equation}
Let us derive a lower bound for the left-hand side of 
\cref{eq:sigma-require}. 
Fix the constant $\gamma>1$ so that $\log\gamma=-\frac{\epsilon\log\tau_\mathrm{max}}{100000}$. By our assumptions we have $\tau\leq\tau_\mathrm{max}$
and $1-\log_\tau\rho\geq\epsilon $, so we have the lower bound
\[
\frac{\log\frac{\rho}{\tau\gamma^2}}{\log\frac{\gamma}{\tau}}
=
\frac{\log\rho-\log\tau-2\log\gamma}{\log\gamma-\log\tau}
=
\frac{1-\log_\tau\rho+\frac{2\log\gamma}{\log\tau}}%
{1-\frac{\log\gamma}{\log\tau}}
\geq
\frac{\epsilon+\frac{2\log\gamma}{\log\tau_\mathrm{max}}}%
{1-\frac{\log\gamma}{\log\tau_\mathrm{max}}}
> 0.99\epsilon\,.
\]
Thus, \cref{eq:sigma-require} holds for all large enough $n$ 
when we require
\[
0.99\epsilon
\geq 
\frac{4C\sigma}{(1-\sigma)\alpha-\delta}\,.
\]
Since $\alpha\epsilon<1$, we have that \cref{eq:sigma-require}
holds when we set
\[
\sigma
=\frac{0.99\epsilon(\alpha-\delta)}{4C+1}
\leq\frac{0.99\epsilon(\alpha-\delta)}{4C+0.99\alpha\epsilon}\,.
\]
We also observe that $(1-\sigma)\alpha-\delta>0$, or equivalently,
$\sigma<(\alpha-\delta)/\alpha$ holds for our choice of $\sigma$.

Having now fixed $\sigma$ and $\gamma$, we observe that in terms of
assumption~2 of the statement of \cref{thm:algorithm}, we have
$\gamma=c_1$ and $\frac{(1-\sigma)\alpha-\delta}{C}=c_2$.  Thus the
assumption $\tau \ge c_1 n^{-c_2}$ guarantees that the right-hand side
of \cref{eq:p-lower} is at least~$1$, which was required for the
existence of $p$.  This completes the parameterization of the
algorithm.

\subsection{Running time}

Let us now analyze the running time of the algorithm. The first
and second steps run in time $\tilde O(nD)$ since $p=O(\log n)$ 
by \cref{eq:p-lower} and $f$ originates from 
\cref{thm:main} and hence is explicit.
From \cref{eq:D-upper} and $n^{\sigma}-1<s$, we have
$nD\leq 4n^{1+(1-\sigma)\alpha}\leq 4n^{2-\sigma}$.
Since \cref{eq:D-upper} holds, the third step of the algorithm
runs in time $O\bigl((n/s)^{2+\eta}\bigr)$ for any constant 
$\eta>0$ that we are free to choose. 
Since $n/s\leq 2n^{1-\sigma}$ for all large enough $n$, we can 
choose $\eta>0$ so that $(2+\eta)(1-\sigma)\leq 2-\sigma$. 
Thus, the first, second, and third steps together run in 
time $O(n^{2-\sigma})$.
The fourth step runs in time $O(n^{2-\sigma}+qs^2d)$. 
Indeed, observe from Claim~1 in \S \ref{subsec:correctness}  that
the detection inequality \cref{eq:detect} holds for at most $q$ entries 
in $Z$. We have $qs^2d\leq qn^{2\sigma+\delta}$, which completes 
the running time analysis and the proof of 
\cref{thm:algorithm}. $\qed$

\section{Applications}
\label{sec:applications}

This section proves \cref{cor:lightbulb,cor:parities}.

\subsection{The light bulb problem}

A useful variant of the \cref{prob:main} asks for all outlier 
pairs of distinct vectors drawn from a \emph{single} set 
$S\subseteq\{-1,1\}^d$ rather than two sets $X,Y$. 
We observe that the single-set variant reduces to 
$\lceil\log |S|\rceil$ instances of the two-set variant by 
numbering the vectors in $S$ with binary numbers from $0$ to $|S|-1$ 
and splitting $S$ into two sets $X_i,Y_i$ based on the value 
of the $i^\mathrm{th}$ bit for each 
$i=0,1,\ldots,\lceil\log |S|\rceil-1$.

\begin{proof}[Proof of \cref{cor:lightbulb}]
We reduce to (the single-set version of) \cref{prob:main} and 
apply \cref{thm:algorithm}. Towards this end, in 
\cref{thm:algorithm} set $\epsilon = 1-1/\kappa$ and 
$\tau_{\mathrm{max}} = \rho_{\mathrm{max}}^\kappa$. 
Suppose we are given an instance of \cref{prob:lightbulb}
whose parameters $n,d,\rho$ satisfy the constraints. 
Set $\tau = \rho^\kappa$. We observe that the constraints in 
\cref{thm:algorithm} are satisfied since 
(i) $d\leq n^\delta$ holds by assumption, 
(ii) $\tau \leq \tau_{\max}$ holds since $\tau = \rho^\kappa \leq
\rho_{\max}^\kappa$, 
(iii) the constants $c_1$ and $c_2$ here match those
in \cref{thm:algorithm}, and the constraint $c_1n^{-c_2/\kappa} \le \rho$
implies $c_1n^{-c_2} \le \tau$,
and (iv)
$\log_\tau\rho = \frac{\log\rho}{\log\tau} =
 \frac{\log\rho}{\log\rho^\kappa} = 1/\kappa \leq 1-\epsilon$.

We claim that $q=1$ in almost all instances
of \cref{prob:lightbulb} whose parameters satisfy the constraints
in Corollary~\ref{cor:lightbulb}.
Indeed, by the Hoeffding bound~\cref{eq:hoeffding} and the union bound, 
the probability that some other pair than the planted pair in an instance
has inner product that exceeds $\tau d$ in absolute value is at most
\[
2n^2\exp\bigl(-\tau^2d/2\bigr)
\leq 2n^2\exp\bigl(-\rho^{2\kappa}\cdot 5\rho^{-2\kappa}\log n\bigr)
= 2n^{-1/2}\,,
\]
so $q=1$ with high probability as $n$ increases.
The claimed running time follows by substituting the chosen constants
and $q=1$ to \cref{eq:runtime}.
\end{proof}

\subsection{Learning parities with noise}

We now generalize the result for parity functions of larger constant
weight, and prove \cref{cor:parities}.

\begin{proof}[Proof of \cref{cor:parities}]

  Fix the constants $0 < \delta < \alpha$, $C>60$, $\xi>1$, $0<\theta < 1$. 
  We will fix the value of the constant $k_0$ later.
  Let $k\geq k_0$ be a constant.
  The algorithm first draws $d$ examples from a given instance 
  of \cref{prob:paritywithnoise} and then transforms these to two
  collections of vectors that we feed to the algorithm of \cref{thm:algorithm}
  and then proceed to mimic the proof of \cref{cor:lightbulb}. 

  Let us first set up some notation.
  For $A,B\subseteq [v]$, let $A\bigtriangleup B = (A\setminus B) \cup
  (B\setminus A)$ denote the symmetric difference of $A$ and $B$. 
  Let $x = (x(1),x(2),\ldots, x(v))\in\{-1,1\}^v$ be a Boolean
  $n$-vector. Let $x^A = \prod_{\ell \in A} x(\ell)$ be the product of
  elements indexed by $A$, with $x^\emptyset = 1$. Observe that
  $x^Ax^B = \prod_{i\in A}x(i)\prod_{j\in B} x(j) = \prod_{\ell \in A\bigtriangleup B} x(\ell) = x^{A\bigtriangleup B}$.
  Let us write $\binom{[n]}{v}$ for the set of all $k$-subsets of $[v]$.

  Suppose we are now given as input an instance of \cref{prob:paritywithnoise}
  with noise level $\eta$ that satisfies $|1-2\eta|\leq\theta<1$. Furthermore,
  we assume that $\eta$ is part of the input. 
  (If this is not the case, at the cost of increasing time complexity, 
  we can search for $\eta$ using a geometric progression with limit $1/2$.)
  With the objective of eventually applying \cref{thm:algorithm}, 
  set 
  \begin{equation}
  \label{eq:parity-rho-choice}
    \rho = |1-2\eta|^\xi
  \end{equation}
  and
  \begin{equation}
  \label{eq:parity-tau-choice}
    \tau = \rho^\xi = |1-2\eta|^{\xi^2}\,.
  \end{equation}
  In particular, we have $\tau<\rho$ since $0<|1-2\eta|<1$ and $\xi>1$. 
  Let $d$ be the least positive integer that satisfies 
  \begin{equation}
    \label{eq:parity-d-choice}
    d \geq (2k+1+4k\zeta)\tau^{-2}(|1-2\eta|-\rho)^{-2}\log v\, ,
  \end{equation}
  where $0 < \zeta < 1/2$ is constant whose value we will fix later.
  Draw from the given instance $d$ example--label pairs 
  $(x_i,y_i) \in$ $\{-1,1\}^v$ $\times$ $\{-1,1\}$ with $i=1,2,\ldots,d$.
  We use these examples to define two collections $X,Y\subseteq\{-1,1\}^d$ 
  of vectors of sizes 
  $\binom{v}{\lfloor k/2 \rfloor}$ and $\binom{v}{\lceil k/2 \rceil}$,
  respectively.
  For all $k\geq \lceil 1/(2\zeta)\rceil$ and all $v\geq 2k$ it is 
  immediate that we have
  \[
    \binom{v}{\lfloor k/2 \rfloor}\leq
    \binom{v}{\lceil k/2 \rceil}\leq
    v^{k(1/2+\zeta)}\,.
  \]
  In particular, we can assume that $|X|,|Y|\leq n$ for 
  $n=\lfloor v^{k(1/2+\zeta)}\rfloor$.

  The set $X$ consists of all the vectors 
  \[
  a^{J_1} = (a_1^{J_1},a_2^{J_1},\ldots,a_d^{J_1}) \in
  \{-1,1\}^d
  \]
  with $a_i^{J_1} = x_i^{J_1}$ for all $i=1,2,\ldots,d$ and 
  $J_1\in\binom{[v]}{\lfloor k/2 \rfloor}$.
  The set $Y$ consists of all the vectors 
  \[
  b^{J_2} =
  (b_1^{J_2},b_2^{J_2},\ldots,b_d^{J_2})
  \]
  with $b_i^{J_1} = x_i^{J_2}y_i$ for all $i=1,2,\ldots,d$ 
  and $J_2\in\binom{[v]}{\lceil k/2 \rceil}$. 

  Let us now study the distribution of inner products between vectors
  in $X$ and~$Y$. 
  We write $\Bin_{\pm1}(d,\beta)$ for a random
  variable that is the sum of $d$ independent random
  variables, each of which takes the value $-1$ with probability $\beta$, and
  the value~$1$ otherwise. Observe that the expectation of
  $\Bin_{\pm1}(d,\beta)$ is $(1-2\beta)d$.

  Let $S\subseteq [v]$ with $|S|=k$ be the support of the parity 
  function that is unknown to us. 
  Recall that $y_i = z_ix_i^S$ with $z_i\in\{-1,1\}$
  getting value $-1$ with probability $\eta$.
  For all $J_1 \in \binom{[v]}{\lfloor k/2 \rfloor}$ and $J_2 \in \binom{[v]}{\lceil k/2 \rceil}$ we have
  \[
  \bra a^{J_1} , b^{J_2} \ket = \sum_{i=1}^d x_i^{J_1}x_i^{J_2} y_i = \sum_{i=1}^d x_i^{J_1\bigtriangleup J_2} x_i^S z_i
  = \sum_{i=1}^d x_i^{J_1\bigtriangleup J_2 \bigtriangleup S} z_i \, .
  \]
  Now observe that there are two distinct cases: If $J_1\bigtriangleup J_2 \neq S$, then
  \begin{equation}
  \label{eq:parity-no-solution-distrib}
  \bra a^{J_1} , b^{J_2} \ket \sim \Bin_{\pm1} (d,1/2) \, .
  \end{equation}
  If $J_1\bigtriangleup J_2 = S$, then 
  \begin{equation}
  \label{eq:parity-solution-distrib}
    \bra a^{J_1} , b^{J_2} \ket = \sum_{i=1}^d x_i^{J_1\bigtriangleup J_2 \bigtriangleup S} z_i
    = \sum_{i=1}^d z_i \sim \Bin_{\pm1}(d,\eta)\, .
  \end{equation}
  Hence, our task of finding the support $S$ reduces to that of
  locating the inner products with distribution $\Bin_{\pm1}(d,\eta)$
  from among those with $\Bin_{\pm1}(d,1/2)$.

  We now argue that our choices \cref{eq:parity-rho-choice},
  \cref{eq:parity-tau-choice}, and \cref{eq:parity-d-choice} suffice
  for the algorithm in \cref{thm:algorithm} to distinguish between the two 
  cases \cref{eq:parity-no-solution-distrib}
  and \cref{eq:parity-solution-distrib} 
  for almost all draws of the $d$ examples. Here we stress that the
  algorithm is deterministic, the randomness is over the draw of the examples.

  From the perspective of the algorithm in \cref{thm:algorithm},
  it suffices that (a) no pair with \cref{eq:parity-no-solution-distrib} 
  exceeds $\tau d$ in absolute-value inner product, and (b) at least
  one of the at most $k^k=O(1)$ pairs with \cref{eq:parity-solution-distrib} 
  has absolute-value inner product at least $\rho d$. 

  To control (a), from \cref{eq:hoeffding} we observe that 
  \begin{align*}
    \Pr\bigl(|\Bin_{\pm1}(d,1/2)| \geq \tau d\bigr) & \leq
    2 \exp \left(-\frac{\tau^2d}{2}\right) \\
    & \leq 2 \exp \left(-\frac{(2k+1+4k\zeta)\log v}{2(|1-2\eta|-\rho)^2} \right) \\
    & = 2 v^{-(2k+1+4k\zeta)(|1-2\eta|-\rho)^{-2}/2} \, .
  \end{align*}
  Since there are at most $n^2 \leq (v^{k(1/2+\zeta)})^2 = v^{k+2\zeta}$ 
  such pairs, we observe by the union bound that (a) holds with high 
  probability as $v$ increases since
  \begin{equation}
    \label{eq:boundforexcesstaud}
    n^2 \cdot 2 v^{-(2k+1+4k\zeta)(|1-2\eta|-\rho)^{-2}/2}
    \leq 2 v^{-(1/2)(|1-2\eta|-\rho)^{-2}} \, .
  \end{equation}

  To control (b), select any fixed pair with \cref{eq:parity-solution-distrib}.
  From \cref{eq:hoeffding} we have 
  \begin{equation}
    \begin{split}
    & \Pr\bigl(|\Bin_{\pm1}(d,\eta)-(1-2\eta)d| \geq (|1-2\eta|-\rho) d\bigr) \\
    \leq \; & 2\exp\left(-\frac{(|1-2\eta|-\rho)^2d}{2}\right) \\
    \leq \;   & 2\exp\left(-\frac{(2k+1+4k\zeta)\log v}{2\tau^2} \right) \\
    = \; & 2 v^{-\frac{2k+1+4k\zeta}{2\tau^2}}\,.
    \end{split}
  \end{equation}
  Thus, (b) holds with high probability as $v$ increases. 

  It remains to verify the constraints for the parameters
  $n,d,\rho,\tau$ in \cref{thm:algorithm}. 
  Suppressing the constants, our choice of $d$ in \cref{eq:parity-d-choice} 
  is $\Theta(k)\cdot|1-2\eta|^{-\Theta(1)}\cdot\log v$. For
  \cref{thm:algorithm} to apply, this must be bounded from
  above by $n^\delta = v^{\Theta(k)}$, which holds if
  $|1-2\eta|\geq v^{-\Theta(k)}$.  This holds by assumption
  for sufficiently large
  $k$. Select $k_0$ so that this constraint holds and 
  $k_0\geq\lceil 1/(2\zeta)\rceil$. 
  We can choose $\tau_{\max} = \theta$ and $\epsilon = 1-1/\xi$.
  We then have
  $\tau = |1-2\eta|^{\xi^2} < \tau_{\max} < 1$ by assumption, as required.
  Since $n \ge v^{k/2}$, we also have by assumption
  \[
  \tau = |1-2\eta|^{\xi^2}
  \ge {c_1}^{\xi^2} v^{-c_2 k/2}
  \ge c_1 n^{-c_2}
  \]
  as required.  The constants $c_1$ and $c_2$ here match those in
  \cref{thm:algorithm}.  Furthermore by the choice of $\epsilon$ we
  have
  \[
  \frac{\log \rho}{\log\tau} = \frac{\log \rho}{\log \rho^\xi} =
  1/\xi = 1-\epsilon \, ,
  \]
  as required. So the constraints of \cref{thm:algorithm} are satisfied.
  For brevity, let $E =
  \frac{0.99\epsilon(\alpha-\delta)}{4C+1}$ and take $\zeta=E/4$. 
  Thus, we have 
  \begin{equation}
  \label{eq:parity-bound2}
  n^{2-E} \leq \left(v^{k(1/2+\zeta)}\right)^{2-E}
          \leq v^{k(1-0.245025(\alpha-\delta)^2(1-1/\xi)^2(1+4C)^{-2})}\,.
  \end{equation}
  The claimed running time \cref{eq:parity-time-bound} follows by observing 
  that \cref{eq:parity-bound2} subsumes the time it takes to construct 
  the collections $X$ 
  and $Y$ together with the time it takes to search the $q$ pairs of buckets 
  with $q\leq k^k=O(1)$ inside the algorithm of \cref{thm:algorithm}.
  
  Inserting our choices \cref{eq:parity-rho-choice} and 
  \cref{eq:parity-tau-choice} into \cref{eq:parity-d-choice} and 
  approximating upwards with $\zeta \leq 1$ and 
  $|1-2\eta|^{2\xi^2+2} (1-\theta^{\xi-1})^2 \leq \tau^2(|1-2\eta|-\rho)^2$ 
  yields
  \begin{align*}
    d \geq \frac{6k}{|1-2\eta|^{2(\xi^2+1)}(1-\theta^{\xi-1})^2}\log v\,.
  \end{align*}
\end{proof}

\section{Nonconstructive existence and a lower bound}
\label{sec:lower-bounds}

This section shows that nontrivial correlation amplifiers exist 
and establishes a lower bound on the output dimension $D$ of
any correlation amplifier. The former is done by a routine 
application of the Hoeffding bound and the latter by applying 
results of Alon~\cite{Alon2003}.

\subsection{Low-dimensional amplifiers exist}
By combining the Hoeffding bound with the union bound, we observe that
low-dimensional amplifiers exist.


\begin{lemma}[Existence]
\label{lem:existence}
  There exists a correlation amplifier $f: \{-1, 1\}^{d} \rightarrow \{-1, 1\}^{D}$ 
  with parameters $(d, D, p, \tau, \gamma)$ whenever 
  $0 < \tau < 1$, $\gamma > 1$, and $d, p, D$ are positive 
  integers satisfying 
  \begin{equation}
  \label{eq:prob-output-dim}
  D\geq  3d\left( \gamma^{p} - 1\right)^{-2}\left(\frac{\gamma}{\tau}\right)^{2p} \, .
  \end{equation}
\end{lemma}

\begin{proof}
  Let $f: \{-1,1\}^{d} \rightarrow \{-1, 1\}^{D}$ be the function 
  which maps $x$ onto $D$ entries of $x^{\otimes p}$ chosen 
  independently at random. That is, each entry of the vector $f(x)$ 
  is the product of $p$ entries of $x$, chosen independently and 
  uniformly at random. 
  
  Let $x, y \in \{-1, 1\}^{d}$ be a fixed pair of vectors, set $c = D (1-\gamma^{-p})\tau^{p}$, 
  and suppose that the following inequality holds, 
\begin{equation}
\label{eq:lower-bound-prob}
  \left| \langle f(x), f(y) \rangle - D \left( \frac{\langle x, y \rangle}{d}\right)^{p} \right| 
  \leq c \,.
\end{equation}
  Observe that if $\left| \langle x, y\rangle \right|< \tau d$ then \cref{eq:lower-bound-prob} implies 
\begin{align*} 
\left| \langle f(x), f(y) \rangle\right| & \leq 
D \left( \frac{\langle x, y \rangle}{d}\right)^{p} + D (1-\gamma^{-p})\tau^{p} \\
& \leq  D \tau^{p} + D (1-\gamma^{-p})\tau^{p} \\
& \leq   (\tau \gamma)^{p} D \,.
\end{align*}
The final inequality holds because $2 - \gamma^{-p} \leq \gamma^{p}$ is 
logically equivalent to  $(\gamma^{p} - 1)^{2} \geq 0$. Similarly, if 
$\left| \langle x, y\rangle \right| \geq \tau d$ then 
\cref{eq:lower-bound-prob} implies the following upper bound,
\begin{align*} 
\langle f(x), f(y) \rangle & \leq 
D \left( \frac{\langle x, y \rangle}{d}\right)^{p} + D (1-\gamma^{-p})\tau^{p} \\
& \leq 
D \left( \frac{\langle x, y \rangle}{d}\right)^{p} + 
D (1-\gamma^{-p})\left( \frac{\langle x, y \rangle}{d}\right)^{p}\\
& \leq  \left(\frac{\gamma\langle x, y \rangle}{d}\right)^{p}D\,.
\end{align*} 
We also obtain a lower bound from \cref{eq:lower-bound-prob} 
when $\left|\langle x, y\rangle\right| \geq \tau d$,
\begin{align*} 
\langle f(x), f(y) \rangle & \geq  
D \left( \frac{\langle x, y \rangle}{d}\right)^{p} - D (1-\gamma^{-p})\tau^{p} \\
& \geq  
D \left( \frac{\langle x, y \rangle}{d}\right)^{p} - 
D (1-\gamma^{-p})\left( \frac{\langle x, y \rangle}{d}\right)^{p}\\
& \geq  \left( \frac{\langle x, y \rangle}{\gamma d}\right)^{p}D\,.
\end{align*} 
In fact, \cref{eq:lower-bound-prob} implies conditions 
\cref{eq:coramp1} and \cref{eq:coramp2} in \cref{def:ca}. 
So if the function $f$ satisfies \cref{eq:lower-bound-prob} for all $x, y \in \{-1, 1\}^{d}$, 
  then $f$ is a correlation amplifier. We use \cref{thm:Hoeffding}
  to bound the probability that \cref{eq:lower-bound-prob} fails, and take a union bound 
  over the range of $f$ to establish a non-constructive existence result for sufficiently large $D$.

  Define the random variable $Z_{f} = \langle f(x), f(y)\rangle$. Since $f(x)$ is a 
  restriction onto $D$ entries of $x^{\otimes p}$ chosen uniformly 
  at random, we have 
  \[E[Z_{f}] = D \left( \frac{\langle x, y\rangle}{d}\right)^{p}\,.\] 
  Observe that $Z_{f} = \sum_{i=1}^{D} Z_{f,i}$ where $Z_{f,i}$ is the 
  product of the $i^{\textrm{th}}$ entries of $f(x)$ and $f(y)$. 
  In particular, $-1 \leq Z_{f,i} \leq 1$ holds for $i = 1, 2, \ldots, D$.
  Summing over the $Z_{f,i}$ in \cref{eq:hoeffding}, the probability that 
  \cref{eq:lower-bound-prob} fails to hold is bounded above by 
 \[
  \textrm{Pr} \left( Z_{f} - \mathrm E[Z_{f}] \geq c \right) 
  \leq e^{-\frac{c^2}{2D} } \,.
\]
  Taking a union bound over all $x, y \in \{-1, 1\}^{d}$, there exists a correlation 
  amplifier with parameters $(d, D, p, \tau, \gamma)$ whenever 
  \[ 2^{2d} e^{-\frac{c^2}{2D}}  < 1\,.\]
  Solving for $D$, we get
  \[
  D\geq \frac{d\ln16}{\tau^{2p}\left( 1- \gamma^{-p}\right)^2} \,.
  \]
  Simplifying this expression and approximating $\ln16$ by $3$ completes the proof.
\end{proof}

\subsection{Lower bound on output dimension}
We next show a lower bound on the output dimension~$D$ of any
correlation amplifier, when the other parameters $d$, $p$, $\tau$ and
$\gamma$ are given.  The proof is based on taking a collection
of $N$ vectors $x_i \in \{-1,1\}^d$, with all pairs below the background
threshold $\tau$, and then bounding the number of their images $f(x_i)
\in \{-1,1\}^D$, whose absolute pairwise correlations are required to be below
$\epsilon = (\tau \gamma)^p$ by \cref{def:ca}.

\begin{lemma}
  \label{lemma:inputside}
  There is a collection of $N=\exp(\tau^2 d/4)$ vectors
  $x_1,x_2,\ldots,x_N$ $\in \{-1,1\}^d$ such that $|\bra x_i, x_j \ket|
  < \tau d$ for all $i \ne j$.
\end{lemma}
\begin{proof}
  We show this by the probabilistic argument.  We call a pair of
  vectors bad if $|\bra x_i, x_j \ket| \ge \tau d$.  Let a collection
  of vectors $X_1,X_2,\ldots,X_N$ be chosen uniformly at random from
  $\{-1,1\}^d$.  Consider a pair $X_i,X_j$ with $i \ne j$, and let $Z_{ij}
  = \bra X_i,X_j \ket$.  Now $Z_{ij}$ is a sum of $d$ independent random
  variables in $[-1, 1]$, with $\textrm E[Z_{ij}] = 0$.  Applying the
  two-sided Hoeffding bound with $c=\tau d$, we observe that the pair
  $X_i,X_j$ is bad with probability
  \[
  \textrm{Pr}(|\bra X_i,X_j \ket| \ge \tau d) =
  \textrm{Pr}(|Z_{ij} - \textrm E[Z_{ij}]| \ge \tau d) \le 2
  \exp(-\tau^2 d/2).
  \]
  Since there are less than $N^2/2 = (1/2)\exp(\tau^2 d/2)$ pairs of
  vectors, the expected number of bad pairs is less than $1$.  Thus in
  at least one collection there are no bad pairs.
\end{proof}

To bound the number of the image vectors, we use a combinatorial
result from Alon~\cite{Alon2003} to bound the rank of their
correlation matrix.

\begin{lemma}
  \label{lemma:alon91}
  Let $A = (a_{ij})$ be an $N \times N$ real, symmetric matrix with
  $a_{ii}=1$ and $|a_{ij}| \le 1/\sqrt{N}$ for all $i \ne j$.  Then
  $\rank(A) \ge N/2$.
\end{lemma}

\begin{proof}
  Apply Alon's Lemma~9.1~\cite{Alon2003} with
  $\epsilon=1/\sqrt{N}$.
\end{proof}

\begin{lemma}
  \label{lemma:alon92}
  Let $B=(b_{ij})$ be an $N \times N$ matrix with $\rank(B)=D'$, and
  let $A = (b_{ij}^k)$, where $k$ is a positive integer.  Then
  $\rank(A) \le \binom{D'+k-1}{k}$.
\end{lemma}

\begin{proof}
  Apply Alon's Lemma~9.2~\cite{Alon2003} with the polynomial
  $P(x)=x^k$.
\end{proof}

The next lemma is in essence Alon's Theorem~9.3~\cite{Alon2003},
modified to avoid any asymptotic notation.  All logarithms here are in
base~$2$.

\begin{lemma}
  \label{lemma:alon93}
  Let $B=(b_{ij})$ be an $N \times N$ real, symmetric matrix with
  $b_{ii}=1$ and $|b_{ij}| \le \epsilon$ for all $i \ne j$, where
  $1/\sqrt{N} \le \epsilon \le 1/100$, and $\rank(B)=D'$.  Then
  \begin{equation}
    \label{eq:Dmin}
    D' \ge \biggl(\frac{r}{5}\biggr) \biggl(\frac{1}{\epsilon}\biggr)^{2r/(r+1)}
  \end{equation}
  where $r = (\log N)/(2 \log (1/\epsilon))$.
\end{lemma}

\begin{proof}
  Choose $r$ as stated.  Note that by the assumed range of $\epsilon$,
  we have $r \ge 1$.  Let further $k = \lceil r \rceil$, so in
  particular $1 \le r \le k < r+1$.

  Let $A = (a_{ij}) = (b_{ij}^k)$.  Since the off-diagonal elements
    of $B$ satisfy $|b_{ij}|<\epsilon$, it follows from the choice of
    $k$ that the off-diagonal elements of $A$ satisfy
    $|a_{ij}| \le \epsilon^k \le \epsilon^r = 1/\sqrt{N}$.
  Combining \cref{lemma:alon91} and \cref{lemma:alon92}, we
  have
  \[
  N/2 \le \rank(A) \le \binom{D'+k-1}{k}
      \le \left(\frac{e(D'+k-1)}{k}\right)^k
      \le \left(\frac{e(D'+r)}{r}\right)^{r+1}.
  \]
  Taking logarithms and rearranging the inequality we obtain
  \[
  \log\left(1+\frac{D'}{r}\right)
  \ge \frac{\log (N/2)}{r+1} - \log e
  \ge \frac{\log N}{r+1} - 2,
  \]
  implying
  \[
  1+\frac{D'}{r} \ge \frac{2^{(\log N) / (r+1)}}{4}.
  \]
  Observing that $\log N = r \log(1/\epsilon^2)$, we get
  \[
  1+\frac{D'}{r} \ge \frac{1}{4} \biggl(\frac{1}{\epsilon}\biggr)^{2r/(r+1)}
  \]
  and, since $\epsilon\le 1/100$ and $r\ge 1$, this implies
  \[
  D' \ge \biggl(\frac{r}{5}\biggr) \biggl(\frac{1}{\epsilon}\biggr)^{2r/(r+1)}
  \]
  as stated. 
\end{proof}

\emph{Remark.}
  The parameter $r$ measures, in a sense, the distance from the case of
  an extremely low correlation requirement $\epsilon=1/\sqrt{N}$.  If
  $r$ tends to infinity, the exponent $2r/(r+1)$ approaches~$2$,
  matching the asymptotic form given by Alon~\cite{Alon2003}.  However,
  with small $r$ the exponent diminishes, reaching $1$ in the limiting
  case $r=1$, that is, when $\epsilon=1/\sqrt{N}$.  In the limiting case
  a direct application of \cref{lemma:alon91} would give the better
  linear bound $D' \ge N/2$.

We can now combine \cref{lemma:inputside}
and~\ref{lemma:alon93} to get a lower bound on output dimension.

\begin{lemma}[Lower bound on output dimension]
  \label{lemma:lowerbound}
  The output dimension of a correlation amplifier with parameters $(d,D,p,\tau,\gamma)$ is
  bounded by
  \[
    D \ge \frac{1}{5} \biggl(\frac{1}{\gamma\tau}\biggr)^p
  \]
  when $(\gamma\tau)^p \le 1/100$ and $p \le \frac{(\log
    e)\tau^2d}{8\log(\frac{1}{\gamma\tau})} $.
\end{lemma}

\begin{proof}
  By \cref{lemma:inputside} there is a collection of
  $N=\exp(\tau^2 d/4)$ vectors $x_1,x_2,\ldots,$ $x_N$ $\in \{-1,1\}^d$
  with correlations below $\tau$ in absolute value.  By Definition~\ref{def:ca}
  their images $u_i = f(x_i) \in \{-1,1\}^D$ have
  correlations below $\epsilon = (\gamma\tau)^p$ in absolute value.

  Consider the $N \times N$ correlation matrix $B = (b_{ij}) = (\bra
  u_i, u_j \ket/D)$.  It is real and symmetric, with diagonal elements
  $b_{ii}=1$ and off-diagonals satisfying $|b_{ij}| \le \epsilon$.  We
  observe that $D' = \rank(B) \le D$.  Applying
  \cref{lemma:alon93} we have
  \[
  r = \frac{\log N}{2 \log(1/\epsilon)}
  = \frac{(\log e) \tau^2 d}{8 p \log(\frac{1}{\gamma\tau})}
  \ge 1,
  \]
  and
  \begin{equation}
    D \ge D'
      \ge \biggl(\frac{r}{5}\biggr) \biggl(\frac{1}{\epsilon}\biggr)^{2r/(r+1)}
      \ge \frac{1}{5} \biggl(\frac{1}{\gamma\tau}\biggr)^p
    \label{eq:Dmincombined}
  \end{equation}
  as claimed.
\end{proof}  

\emph{Remark.}  At the limiting case where $p = \frac{(\log
  e)\tau^2d}{8\log(\frac{1}{\gamma\tau})} $, we have $r=1$ and
$\epsilon = 1/\sqrt{N} = \exp(-t^2 d / 8)$, and the bound
\cref{eq:Dmincombined} becomes $D \ge \exp(\tau^2 d/8)$.  For $p$
greater than the limit, one can essentially map all of the
$N=\exp(\tau^2 d/4)$ input vectors to \emph{orthogonal} output vectors
of dimension $D \le 2N$ using a Hadamard matrix, in which case
\cref{eq:coramp1} holds for arbitrary $p>1$.


\appendix

\section{An expander family}
\label{appendix:expanders}

This section proves \cref{lem:expander-family} following 
Reingold, Vadhan and Wigderson~\cite{ReingoldVadhanWigderson};
we present the proof for completeness of exposition only with 
no claim of originality. Following Reingold, Vadhan and 
Wigderson~\cite{ReingoldVadhanWigderson} we will work with
\textit{normalized} eigenvalues. To avoid confusion with the
unnormalized treatment in the manuscript proper, we say that
a graph is a $[D,\Delta,\lambda]$-\emph{graph} if the
graph has $D$ vertices, is $\Delta$-regular, and 
$|\lambda_2|/\Delta\leq\lambda$. (Here $|\lambda_2|$ is the
unnormalized second eigenvalue as defined in the manuscript proper.)

We refer to Sections 2.3 and 3.1 
of Reingold, Vadhan, and Wigderson~\cite{ReingoldVadhanWigderson}
for the definition of the square $G^2$ of a graph $G$, 
the tensor product $G_1\otimes G_2$ of graphs $G_1,G_2$, and 
the zigzag product $G\zigzag H$ of graphs $G,H$.
The following omnibus result collects elements of 
Propositions 2.3, Proposition 2.4, Theorem 3.2 and Theorem 4.3 of 
\cite{ReingoldVadhanWigderson} which will be sufficient to control
the second normalized eigenvalue for our present purposes.
(We choose to omit the details of the rotation maps with the 
understanding that they can be found 
in \cite{ReingoldVadhanWigderson}.)

\begin{lemma}%
[Reingold, Vadhan, and Wigderson~\cite{ReingoldVadhanWigderson}]
\label{lem:eigenvalue-bounds}
The following bounds hold.
\begin{enumerate}
\item 
If $G$ is a $[D,\Delta,\lambda]$-graph, 
then $G^{2}$ is a $[D,\Delta^{2},\lambda^{2}]$-graph.
\item 
If $G_1$ is a $[D_1,\Delta_1,\lambda_1]$-graph and 
$G_2$ is a $[ D_2,\Delta_2,\lambda_2]$-graph,\\
then $G_1\otimes G_2$ is a 
$[D_1D_2,\Delta_1\Delta_2,\max(\lambda_1,\lambda_2)]$-graph.
\item 
If $G$ is a $[D_1,\Delta_1,\lambda_1]$-graph and 
$H$ a $[\Delta_1,\Delta_2,\lambda_2]$-graph,\\ 
then $G \zigzag H$ is a 
$[D_1\Delta_1,\Delta_2^2,f(\lambda_1,\lambda_2)]$-graph with
\[
f(\lambda_1,\lambda_2) = 
\frac{1}{2}\left(1-\lambda_2^2\right)\lambda_1 
+ \frac{1}{2}\sqrt{\left(1-\lambda_2^2\right)^2\lambda_1^2 
+ 4\lambda_2^2} 
\leq \lambda_1 + \lambda_2\,.
\]
\end{enumerate}
\end{lemma}

Let us study the following sequence of graphs.
Let $H$ be a $[D,\Delta,\lambda]$-graph. 
Let $G_1=H^2$, $G_2 = H\otimes H$, and 
for $t=3,4,\ldots$ let 
\begin{equation}
\label{eq:g-t}
G_t = 
\left( G_{\lceil \frac{t-1}{2}\rceil} \otimes G_{\lfloor \frac{t-1}{2} \rfloor} \right)^{2} \zigzag H \,.
\end{equation}
From \cref{lem:eigenvalue-bounds} it 
is easily seen that $G_t$ is a $[D^t,\Delta^2,\lambda_t]$-graph 
with $\lambda_t$ defined by 
\begin{alignat*}{2}
\lambda_1 &= \lambda^2\,,&&\\
\lambda_2 &= \lambda\,,&&\\
\lambda_{2t-1} &= \lambda + \lambda_{t-1}^2\,,&\qquad \text{for }t&=2,3\ldots\,,\text{ and}\\
\lambda_{2t} &= 
  \max(\lambda + \lambda_t^2, \lambda + \lambda_{t-1}^2)\,,&\qquad
  \text{for }t&=2,3,\ldots\,.
\end{alignat*}
\begin{lemma}[Reingold, Vadhan, and 
Wigderson~{\cite[Theorem~3.3]{ReingoldVadhanWigderson}}]
\label{lem:gt-rot}
The rotation map $\Rot_{G_t}$ can be computed in time
$\poly(t,\log D)$ and by making $\poly(t)$ evaluations 
of $\Rot_H$.
\end{lemma}

\begin{lemma}\label{lem:eigenvalue-bound}
If\/ $0\leq\lambda \leq 1/4$ then 
$\lambda_t \leq \lambda + 4\lambda^2$ for all $t\geq 1$. 
\end{lemma}

\begin{proof}
The conclusion is immediate for $t\leq 2$. 
So suppose that the conclusion holds up to $2t-2$. 
We need to show that the conclusion holds for 
$\lambda_{2t-1}$ and $\lambda_{2t}$. 
By induction, it suffices to show that 
\[ 
\lambda_{2t-1}
\leq \lambda + (\lambda + 4\lambda^2)^2 
\leq \lambda + 4\lambda^{2}\,.
\]
Observing that 
$\lambda^2 + 8\lambda^3 +16\lambda^4 \leq 4 \lambda^2$
holds for $0\leq\lambda\leq 1/4$
yields the desired conclusion. The proof for $\lambda_{2t}$ 
is identical. 
\end{proof}

Finally, we construct the expanders that we require in the 
manuscript proper.
\begin{lemma}%
[\cref{lem:expander-family} stated with normalized eigenvalue notation]
For all integers $t\geq 1$ and $b\geq 10$ there exists 
a $[2^{16bt},2^{4b},16\cdot 2^{-b}]$-graph whose rotation map
can be evaluated in time $\poly(b,t)$.
\end{lemma}

\begin{proof}
Take $q=2^b$ and $d=15$ in 
Proposition 5.3 of Reingold, Vadhan, and 
Wigderson~\cite{ReingoldVadhanWigderson}
to obtain a $[2^{16b},2^{2b},15\cdot 2^{-b}]$-graph $H$
whose rotation map can be computed in time $\poly(b)$. 
(Indeed, observe that an irreducible polynomial to perform
the required arithmetic in the finite field of order $2^b$ 
can be constructed in deterministic time 
$\poly(b)$ by an algorithm of Shoup~\cite{Shoup1990}.)
Let us study the sequence $G_t$ given by \cref{eq:g-t}. 
The time complexity of the rotation map follows immediately 
from \cref{lem:gt-rot}.
Since $b\geq 10$, \cref{lem:eigenvalue-bound} gives that  
$\lambda_t\leq\lambda + 4\lambda^{2}$ for all $t \geq 1$. 
Take $\lambda=15\cdot 2^{-b}$ and observe that since 
$b\geq 10$ we have $2^{-b}<1/900$. Thus,
$\lambda_t\leq 15\cdot 2^{-b} + 4 (15 \cdot 2^{-b})^2 
= 15 \cdot 2^{-b} + 900 \cdot 2^{-2b} 
\leq 16 \cdot 2^{-b}$.
\end{proof}


\bibliographystyle{plainurl}
\bibliography{refs}

\begin{thebibliography}{10}

\bibitem{Ahle2015}
Thomas~D. Ahle, Rasmus Pagh, Ilya Razenshteyn, and Francesco Silvestri.
\newblock On the complexity of inner product similarity join.
\newblock {\em arXiv}, abs/1510.02824, 2015.

\bibitem{Alman2016}
Josh Alman, Timothy~M. Chan, and Ryan Williams.
\newblock Polynomial representations of threshold functions and algorithmic
  applications.
\newblock {\em arXiv}, abs/1608.04355, 2016.

\bibitem{Alman2015}
Josh Alman and Ryan Williams.
\newblock Probabilistic polynomials and {H}amming nearest neighbors.
\newblock In {\em Proc. 56th Annual {IEEE} Symposium on Foundations of Computer
  Science (FOCS)}, pages 136--150, Los Alamitos, CA, USA, 2015. IEEE Computer
  Society.

\bibitem{Alon2003}
Noga Alon.
\newblock Problems and results in extremal combinatorics -- {I}.
\newblock {\em Discrete Math.}, 273(1-3):31--53, 2003.

\bibitem{Andoni2014}
Alexandr Andoni, Piotr Indyk, Huy~L. Nguyen, and Ilya Razenshteyn.
\newblock Beyond locality-sensitive hashing.
\newblock In {\em Proc. 25th Annual {ACM-SIAM} Symposium on Discrete Algorithms
  (SODA)}, pages 1018--1028, Philadelphia, PA, USA, 2014. Society for
  Industrial and Applied Mathematics.

\bibitem{Andoni2016}
Alexandr Andoni, Thijs Laarhoven, Ilya~P. Razenshteyn, and Erik Waingarten.
\newblock Optimal hashing-based time-space trade-offs for approximate near
  neighbors.
\newblock {\em arXiv}, abs/1608.03580, 2016.

\bibitem{Andoni2015}
Alexandr Andoni and Ilya Razenshteyn.
\newblock Optimal data-dependent hashing for approximate near neighbors.
\newblock In {\em Proc. 47th {ACM} Annual Symposium on the Theory of Computing
  ({STOC})}, pages 793--801, New York, NY, USA, 2015. Association for Computing
  Machinery.

\bibitem{Blum2003}
Avrim Blum, Adam Kalai, and Hal Wasserman.
\newblock Noise-tolerant learning, the parity problem, and the statistical
  query model.
\newblock {\em J. {ACM}}, 50(4):506--519, 2003.

\bibitem{Celis2013}
L.~Elisa Celis, Omer Reingold, Gil Segev, and Udi Wieder.
\newblock Balls and bins: Smaller hash families and faster evaluation.
\newblock {\em {SIAM} J. Comput.}, 42(3):1030--1050, 2013.

\bibitem{Chan2016}
Timothy~M. Chan and Ryan Williams.
\newblock Deterministic {APSP}, orthogonal vectors, and more: Quickly
  derandomizing {R}azborov-{S}molensky.
\newblock In Robert Krauthgamer, editor, {\em Proc. 27th Annual {ACM-SIAM}
  Symposium on Discrete Algorithms ({SODA})}, pages 1246--1255, Arlington, VA,
  USA, 2016. {Society for Industrial and Applied Mathematics}.

\bibitem{Dubiner2010}
Moshe Dubiner.
\newblock Bucketing coding and information theory for the statistical
  high-dimensional nearest-neighbor problem.
\newblock {\em {IEEE} Trans. Inf. Theory}, 56(8):4166--4179, 2010.

\bibitem{Feldman2009}
Vitaly Feldman, Parikshit Gopalan, Subhash Khot, and Ashok~Kumar Ponnuswami.
\newblock On agnostic learning of parities, monomials, and halfspaces.
\newblock {\em {SIAM} J. Comput.}, 39(2):606--645, 2009.

\bibitem{Gionis1999}
Aristides Gionis, Piotr Indyk, and Rajeev Motwani.
\newblock Similarity search in high dimensions via hashing.
\newblock In Malcolm~P. Atkinson, Maria~E. Orlowska, Patrick Valduriez,
  Stanley~B. Zdonik, and Michael~L. Brodie, editors, {\em Proc. 25th
  International Conference on Very Large Data Bases (VLDB'99)}, pages 518--529,
  Edinburgh, Scotland, {UK}, 1999. Morgan Kaufmann.

\bibitem{Gopalan2015}
Parikshit Gopalan, Daniek Kane, and Raghu Meka.
\newblock Pseudorandomness via the {D}iscrete {F}ourier {T}ransform.
\newblock In {\em Proc. {IEEE} 56th Annual Symposium on Foundations of Computer
  Science ({FOCS})}, pages 903--922, Berkeley, CA, USA, 2015. {IEEE} Computer
  Society.

\bibitem{Gopalan2013}
Parikshit Gopalan, Raghu Meka, Omer Reingold, and David Zuckerman.
\newblock Pseudorandom generators for combinatorial shapes.
\newblock {\em {SIAM} J. Comput.}, 42(3):1051--1076, 2013.

\bibitem{Grigorescu2011}
Elena Grigorescu, Lev Reyzin, and Santosh Vempala.
\newblock On noise-tolerant learning of sparse parities and related problems.
\newblock In {\em Proc. 22nd International Conference on Algorithmic Learning
  Theory (ALT)}, pages 413--424, Berlin, Germany, 2011. Springer.

\bibitem{Hoeffding1963}
Wassily Hoeffding.
\newblock Probability inequalities for sums of bounded random variables.
\newblock {\em J. Amer. Statist. Assoc.}, 58:13--30, 1963.

\bibitem{HooryLinialWigderson}
Shlomo Hoory, Nathan Linial, and Avi Wigderson.
\newblock Expander graphs and their applications.
\newblock {\em Bull. Amer. Math. Soc.}, 43(4):439--561, 2006.

\bibitem{Impagliazzo2001}
Russell Impagliazzo and Ramamohan Paturi.
\newblock On the complexity of $k$-{SAT}.
\newblock {\em J. Comput. Syst. Sci.}, 62(2):367--375, 2001.

\bibitem{Indyk1998}
Piotr Indyk and Rajeev Motwani.
\newblock Approximate nearest neighbors: Towards removing the curse of
  dimensionality.
\newblock In {\em Proc. 30th Annual {ACM} Symposium on the Theory of Computing
  (STOC)}, pages 604--613, New York, NY, USA, 1998. Association for Computing
  Machinery.

\bibitem{Kane2011}
Daniel~M. Kane, Raghu Meka, and Jelani Nelson.
\newblock Almost optimal explicit {J}ohnson-{L}indenstrauss families.
\newblock In {\em Proc. 14th International Workshop on Approximation,
  Randomization, and Combinatorial Optimization, {RANDOM} and 15th
  International Workshop on Algorithms and Techniques, {APPROX}}, pages
  628--639, Princeton, NJ, USA, 2011.

\bibitem{Kapralov2015}
Michael Kapralov.
\newblock Smooth tradeoffs between insert and query complexity in nearest
  neighbor search.
\newblock In {\em Proc. 34th {ACM} Symposium on Principles of Database Systems
  ({PODS})}, pages 329--342, New York, NY, USA, 2015. Association for Computing
  Machinery.

\bibitem{Karppa2016}
Matti Karppa, Petteri Kaski, and Jukka Kohonen.
\newblock A faster subquadratic algorithm for finding outlier correlations.
\newblock In {\em Proc. 27th Annual {ACM-SIAM} Symposium on Discrete
  Algorithms, ({SODA})}, pages 1288--1305, Arlington, VA, USA, 2016. Society
  for Industrial and Applied Mathematics.

\bibitem{Kothari2015}
Pravesh~K. Kothari and Raghu Meka.
\newblock Almost optimal pseudorandom generators for spherical caps.
\newblock In {\em Proc. 47th Annual {ACM} Symposium on Theory of Computing
  ({STOC})}, pages 247--256, Portland, OR, USA, 2015.

\bibitem{LeGall2012}
Fran{\c{c}}ois {Le Gall}.
\newblock Faster algorithms for rectangular matrix multiplication.
\newblock In {\em Proc. 53rd Annual {IEEE} Symposium on Foundations of Computer
  Science ({FOCS})}, pages 514--523, Los Alamitos, CA, USA, 2012. IEEE Computer
  Society.

\bibitem{Lubotzky1988}
A.~Lubotzky, R.~Phillips, and P.~Sarnak.
\newblock Ramanujan graphs.
\newblock {\em Combinatorica}, 8:261--277, 1988.

\bibitem{May2015}
Alexander May and Ilya Ozerov.
\newblock On computing nearest neighbors with applications to decoding of
  binary linear codes.
\newblock In {\em Proc. {EUROCRYPT} 2015 - 34th Annual International Conference
  on the Theory and Applications of Cryptographic Techniques}, pages 203--228,
  Berlin, Germany, 2015. Springer.

\bibitem{Mossel2004}
Elchanan Mossel, Ryan O'Donnell, and Rocco~A. Servedio.
\newblock Learning functions of {$k$} relevant variables.
\newblock {\em J. Comput. Syst. Sci.}, 69(3):421--434, 2004.

\bibitem{Motwani2007}
Rajeev Motwani, Assaf Naor, and Rina Panigrahy.
\newblock Lower bounds on locality sensitive hashing.
\newblock {\em {SIAM} J. Discrete Math.}, 21(4):930--935, 2007.

\bibitem{ODonnell2014}
Ryan O'Donnell, Yi~Wu, and Yuan Zhou.
\newblock Optimal lower bounds for locality-sensitive hashing (except when q is
  tiny).
\newblock {\em {ACM} Trans. Comput. Theory}, 6(1):Article 5, 2014.

\bibitem{Pagh2016}
Rasmus Pagh.
\newblock Locality-sensitive hashing without false negatives.
\newblock In {\em Proc. 27th Annual {ACM-SIAM} Symposium on Discrete Algorithms
  (SODA)}, pages 1--9, Philadelphia, PA, USA, 2016. Society for Industrial and
  Applied Mathematics.

\bibitem{Paturi1989}
Ramamohan Paturi, Sanguthevar Rajasekaran, and John~H. Reif.
\newblock The light bulb problem.
\newblock In {\em Proc. 2nd Annual Workshop on Computational Learning Theory
  ({COLT})}, pages 261--268, New York, NY, USA, 1989. Association for Computing
  Machinery.

\bibitem{Pham2016}
Ninh Pham and Rasmus Pagh.
\newblock Scalability and total recall with fast {C}overing{LSH}.
\newblock {\em ar{X}iv}, abs/1602.02620, 2016.

\bibitem{ReingoldVadhanWigderson}
Omer Reingold, Salil Vadhan, and Avi Wigderson.
\newblock Entropy waves, the zig-zag graph product, and new constant-degree
  expanders.
\newblock {\em Ann. of Math.}, 155(1):157--187, 2002.

\bibitem{Shoup1990}
Victor Shoup.
\newblock New algorithms for finding irreducible polynomials over finite
  fields.
\newblock {\em Math. Comp.}, 54:435--447, 1990.

\bibitem{Valiant2015}
Gregory Valiant.
\newblock Finding correlations in subquadratic time, with applications to
  learning parities and the closest pair problem.
\newblock {\em J. {ACM}}, 62(2):Article 13, 2015.

\bibitem{Valiant1988}
Leslie~G. Valiant.
\newblock Functionality in neural nets.
\newblock In {\em Proc. 1st Annual Workshop on Computational Learning Theory
  ({COLT})}, pages 28--39, New York, NY, USA, 1988. Association for Computing
  Machinery.

\end{thebibliography}

\end{document}